\documentclass[final,leqno]{siamltex}

\usepackage[usenames,dvipsnames]{xcolor}
\usepackage{amsfonts}
\usepackage[sort]{cite}
\usepackage{amsmath} 
\usepackage{amssymb}  
\usepackage{graphicx}
\usepackage{epsfig}
\usepackage{epstopdf}
\usepackage[normal]{subfigure}
\usepackage{mathtools}
\usepackage{etoolbox}
\usepackage{lipsum}
\usepackage{caption}
\usepackage{soul}
\usepackage[colorlinks=true]{hyperref}
\usepackage[normal]{subfigure}

\usepackage{algorithm,algorithmic}
\newenvironment{varalgorithm}[1]
  {\algorithm\renewcommand{\thealgorithm}{#1}}
  {\endalgorithm}

\def\cov{\mathop{\mathrm{Cov}}}  
\def\var{\mathop{\mathrm{Var}}}  
\def\diag{\mathop{\mathrm{diag}}}  
\def\trace{\mathop{\mathrm{trace}}}  
\def\IID{\mathop{\mathrm{IID}}}  

\def\nc{\mathop{\mathrm{nc}}}  

\let\bbordermatrix\bordermatrix
\patchcmd{\bbordermatrix}{8.75}{4.75}{}{}
\patchcmd{\bbordermatrix}{\left(}{\left[}{}{}
\patchcmd{\bbordermatrix}{\right)}{\right]}{}{}
\renewcommand{\thefigure}{\thesection.\arabic{figure}}
\renewcommand{\theequation}{\thesection.\arabic{equation}}
\numberwithin{equation}{section}

\hypersetup{
     colorlinks = true,
     linkcolor = [rgb]{0,0,1},
     anchorcolor = [rgb]{0,0,1},
     citecolor = [rgb]{0.9,0.5,0},
     filecolor = [rgb]{0,0,1},
     pagecolor = [rgb]{1,1,0},
     urlcolor = [rgb]{0.9,0.5,0},
     bookmarks,
        bookmarksopen = true,
        bookmarksnumbered = true,
        breaklinks = true,
        linktocpage,
        pagebackref,
        colorlinks = false,
        linkcolor = [rgb]{0.2,0.6,0.2},
        urlcolor  = [rgb]{0,0,1},
        citecolor = [rgb]{0.9,0.5,0},
        anchorcolor = [rgb]{0.2,0.6,0.2},
        hyperindex = true,
        hyperfigures
}

\newcommand{\T}{^{\mbox{\tiny T}}}

\newtheorem{remark}{Remark}
\newtheorem{problem}{Problem}

\newtheorem{example}{Example}

\newcommand{\sr}{\stackrel}

\newcommand{\noi}{\noindent}

\newcommand{\be}{\begin{equation}}
\newcommand{\ee}{\end{equation}}
\newcommand{\bea}{\begin{eqnarray}}
\newcommand{\eea}{\end{eqnarray}}
\newcommand{\bes}{\begin{eqnarray*}}
\newcommand{\ees}{\end{eqnarray*}}

\newcommand{\bfi}{\begin{figure}}
\newcommand{\bfit}{\begin{figure}[t]}
\newcommand{\bfib}{\begin{figure}[b]}
\newcommand{\bfih}{\begin{figure}[h]}
\newcommand{\bfip}{\begin{figure}[p]}
\newcommand{\efi}{\end{figure}}

\newcommand{\bi}{\begin{itemize}}
\newcommand{\ei}{\end{itemize}}
\newcommand{\ben}{\begin{enumerate}}
\newcommand{\een}{\end{enumerate}}

\newcommand{\bp}{\begin{problem}}
\newcommand{\ep}{\end{problem}}
\newcommand{\hso}{\hspace{.1in}}
\newcommand{\hst}{\hspace{.2in}}

\begin{document}

\title{Optimal Estimation via Nonanticipative Rate Distortion Function and Applications to Time-Varying Gauss-Markov Processes\footnote{{Part of the material in this paper was presented in IEEE Conference on Decision and Control (CDC), Las Vegas, NV, USA, 2016 \cite{stavrou-charalambous-charalambous2016cdc}}.}} 

\author{Photios A. Stavrou\thanks{Department of Electronic Systems, Aalborg University, Aalborg, Denmark. ({\tt fos@es.aau.dk}).}
        \and Themistoklis~Charalambous\thanks{Department of Signals and Systems, Chalmers University of Technology, Gothenburg, Sweden. ({\tt themistoklis.charalambous@chalmers.se}).}
\and Charalambos D.~Charalambous\thanks{Department of Electrical and Computer Engineering, University of Cyprus, Nicosia, Cyprus, ({\tt chadcha@ucy.ac.cy}).}
\and Sergey Loyka\thanks{School of Electrical Engineering and Computer Science, University of Ottawa, Ontario, Canada, K1N 6N5. ({\tt sergey.loyka@ottawa.ca}).}}

\maketitle

\begin{abstract}
In this paper, we develop {finite-time horizon} causal filters using the nonanticipative rate distortion theory. We apply the {developed} theory to {design optimal filters for} time-varying multidimensional Gauss-Markov processes, subject to a mean square error fidelity constraint. We show that such filters are equivalent to the design of an optimal \texttt{\{encoder, channel, decoder\}}, which ensures that the error satisfies {a} fidelity constraint.  Moreover, we derive a universal lower bound on the mean square error of any estimator of time-varying multidimensional Gauss-Markov processes in terms of conditional mutual information. Unlike classical Kalman filters, the filter developed is characterized by a reverse-waterfilling algorithm, which ensures {that} the fidelity constraint is satisfied. The theoretical results are demonstrated via illustrative examples.
\end{abstract}

\begin{keywords} 
Causal filters, nonanticipative rate distortion function, mean square error distortion, reverse-water filling, universal lower bound.
\end{keywords}


\pagestyle{myheadings}
\markboth{P. A. STAVROU, T. CHARALAMBOUS, C. D. CHARALAMBOUS AND S. LOYKA}{OPTIMAL ESTIMATION VIA NRDF}

%
%
%
%
\section{Introduction}\label{introduction} 
Motivated by  real-time control applications, of communication system design, Gorbunov and Pinsker in \cite{gorbunov-pinsker1973} introduced the so-called nonanticipatory $\epsilon$-entropy of general processes, (see \cite[Introduction I]{gorbunov-pinsker1973}). The nonanticipative $\epsilon$-entropy is equivalent to Shannon's Rate Distortion Function (RDF) \cite{shannon1959,berger1971} with an additional  causality constraint on the optimal reproduction or estimator.  
%
 \par Along the same lines, for a two-sample Gaussian process, Bucy in \cite{bucy} derived a causal estimator using the Distortion Rate Function\footnote{The DRF is the dual of the RDF (see \cite{cover-thomas2006}).} (DRF)  subject to a causality constraint.  Galdos and Gustafson in \cite{galdos-gustafson1977} applied the classical RDF to design reduced order estimators. Tatikonda, in his Ph.D. thesis \cite{tatikonda2000}, introduced the so-called sequential RDF, which is a variant of the nonanticipatory $\epsilon$-entropy and related this to the Optimal Performance Theoretically Attainable (OPTA) by causal codes, as defined by Neuhoff and Gilbert in \cite{neuhoff-gilbert1982}. Moreover, in \cite{tatikonda2000}, the author computed the sequential RDF of a scalar-valued Gaussian process described by discrete recursion driven by an Independent and Identically Distributed ($\IID$) Gaussian noise process, subject to a Mean Square Error (MSE) fidelity constraint. In  addition, the author of \cite{tatikonda2000} illustrated by construction, how to communicate the Gaussian process, optimally over a memoryless Additive Gaussian Noise (AGN) channel subject to a power constraint, that is, by designing the \texttt{\{encoder, decoder\}} so that the AGN channel operates at its capacity and the sequential RDF is achieved.  In \cite{tatikonda-sahai-mitter2004ieeetac} the authors showed that if the Gaussian process is unstable then sequential RDF is bounded below by the {sum of} logarithm{s} of the absolute value{s} of {the} unstable eigenvalue{s}, and that a necessary condition for asymptotic stability of {a} linear control system over {a} limited-rate communication channel is  ``the capacity of the channel, noiseless or noisy, is larger than the sum of logarithms of the absolute values of the unstable eigenvalues of the open loop control system''. Similar conditions are derived by many authors via alternative methods \cite{wong-brockett1999ieeetac,elia2004ieeetac,charalambous-farhadi2008}.  

\par In \cite{charalambous-stavrou-ahmed2014ieeetac} the authors re-visited the relation between information theory and filtering theory, by introducing the  so-called Nonanticipative RDF (NRDF), and derived existence of optimal solutions. Moreover,  under the assumption that the  solution to the NRDF is  time-invariant, the form of the optimal reproduction distribution is derived. This expression is applied to derive a sub-optimal causal filter for time-invariant multidimensional partially observed Gaussian processes described by discrete-time recursions. For fully observed Gaussian processes the solution given in \cite{charalambous-stavrou-ahmed2014ieeetac} is optimal and generalizes the solution given in \cite{tatikonda-sahai-mitter2004ieeetac} to multidimensional Gaussian processes with MSE distortion instead of per letter distortion. Recently, Stavrou \emph{et al.} in \cite{stavrou-kourtellaris-charalambous2016ieeetit} showed that nonanticipative $\epsilon$-entropy, sequential RDF, and NRDF are equivalent notions. The optimal reproduction  distribution which minimizes directed information from one process to another process subject to average distortion constraint  is given in \cite{charalambous-stavrou2014ecc}. 

\par The NRDF has been  used in many other communication-related problems. For example, Derpich and $\O$stergaard in \cite{derpich-ostergaard2012} applied the nonanticipatory $\epsilon$-entropy of the scalar Gaussian process subject to a MSE fidelity constraint,  to derive several bounds on the OPTA by causal and zero-delay codes. Kourtellaris \emph{et al.} in \cite{kourtellaris-charalambous-boutros2015isit} illustrated the simplicity of jointly design{ing} an \texttt{\{encoder, channel, decoder\}} operating optimally in real-time, for a Binary Symmetric Markov process subject to a Hamming distance distortion function, which is communicated over a finite state channel with unit memory on past channel outputs (with some symmetry) subject to a transmission cost constraint. The NRDF is  {also applied} in control-related problems using zero-delay communication constraints. For example, Tanaka {\it et al.} in \cite{tanaka-kim-parrilo-mitter2016ieeetac} investigated {a} time-varying multidimensional fully observed Gauss-Markov process with {letter-by-letter distortion} motivated by the utility of such communication model in real-time communications for control. In addition, in \cite{tanaka-kim-parrilo-mitter2016ieeetac} the authors apply semidefinite programming to find, numerically, optimal solutions to the sequential RDF (or NRDF) of time-varying fully observed Gauss-Markov sources.   

\subsection{Problem Statement}

In this paper we investigate the following estimation problem: \textit{given an arbitrary random process, we wish to design an optimal  communication system so that at the output of this system the estimated process satisfies {an} end-to-end average fidelity or distortion criterion}.

This problem is  equivalent to the design of an optimal \texttt{\{encoder, decoder\}}, which communicates the arbitrary process and reconstructs it at the output of the decoder. Formally, the problem can be cast as follows:
\begin{problem}\label{fp1}(Information-based estimation)
Given 
\begin{itemize}
\item[(a)] an arbitrary random process $\{X_t:~t=0,\ldots,n\}$ taking values in  complete separable metric spaces $\{ {\cal X}_t: t=0, \ldots, n\}$, with  conditional distribution $\{P_{X_t|X^{t-1}}(dx_t|x^{t-1}):t=0,\ldots,n\}$, $x^{t-1}\triangleq\{x_0,x_1,\ldots,x_{t-1}\}$; 
\item[(b)]  a distortion function or fidelity of reproducing $x_t$ by $y_t \in {\cal Y}_t \subseteq {\cal X}_t, t=0,1, \ldots, n$, defined by a real-valued measurable function $d_{0,n}(\cdot, \cdot)$
\begin{align}
d_{0,n}(x^n,y^n)\triangleq\sum_{t=0}^n\rho_t(T^tx^n, T^ty^n)\in [0, \infty],  \label{introduction:distortion_function}
\end{align}
where $T^tx^n \subseteq \{x_0, x_1,\ldots, x_t\},~T^ty^n \subseteq \{y_0, y_1,\ldots, y_t\}$ is either fixed or non-increasing with time\footnote{For example $\rho_t(T^tx^n, T^ty^n)=\rho(x_t,y_t), t=0, \ldots, n$, where $\rho(\cdot,\cdot)$ is a  distance metric.} for $t=0,1,\ldots,n$,
\end{itemize}
 we wish to determine  an optimal probabilistic \texttt{\{encoder, channel, decoder\}} which communicates $\{X_t:t=0, \ldots, n\}$ and reconstructs it at the output of the decoder or estimator, while it satisfies the end-to-end average fidelity given by 
\begin{align}
\frac{1}{n+1}\mathbb{E}\left\{d_{0,n}(X^n,Y^n)\right\}\leq{D}, \hst \forall D \in [0, \infty).\label{fid_1}
\end{align}
\end{problem}

\par The above definition of estimation problem ensures fidelity (\ref{fid_1}) is satisfied, hence it is fundamentally different from standard approaches of estimation theory, such as, MSE estimation. In general, to achieve such fidelity, for any $D \in [D_{\min}, \infty]$, we know from Shannon's information theory \cite{shannon1959}, that we need to design the actual observation process or sensor from which the estimator is  constructed.  This is equivalent the construction  of the \texttt{\{encoder, channel, decoder\}}, as shown in Fig.~\ref{introduction:fig:realization:nrdf:zero-delay}. 
\begin{figure}[!h]
\centering
\includegraphics[width=\columnwidth]{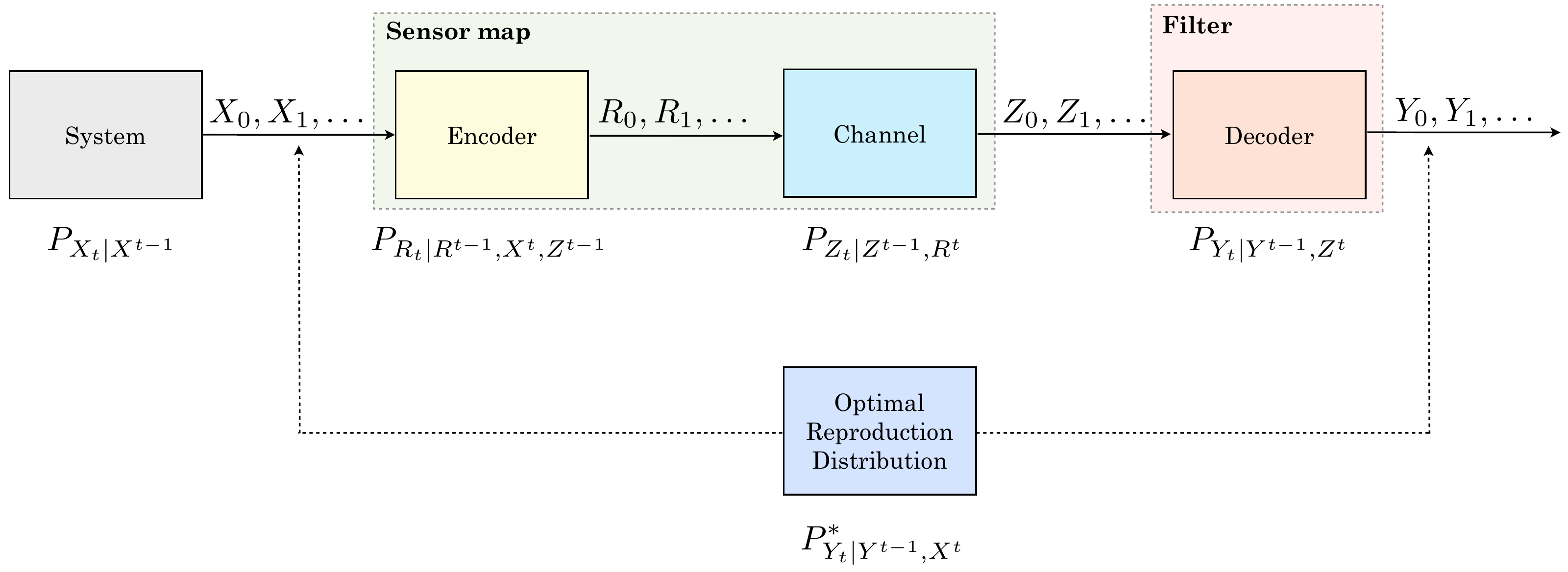}
\caption{Block diagram of Problem~\ref{fp1} with probabilistic \texttt{\{encoder, channel, decoder\}}.  
}
\label{introduction:fig:realization:nrdf:zero-delay}
\end{figure}
This point of view was recognized by Gorbunov and Pinsker \cite{gorbunov-pinsker1973}, and Bucy \cite{bucy} several years ago. 

\noi Our main objective is to address Problem~\ref{fp1} using information-theoretic measures. 
{The natural information-theoretic measure to addresse Problem~\ref{fp1} is the NRDF; this is justified by the equivalence of NRDF and nonanticipatory $\epsilon$-entropy.} 

\par In the next section, we describe the contributions and the fundamental differences between information-based estimation via NRDF and Bayesian estimation theory.
 
\subsection{Relation {between} Bayesian Estimation and Estimation using NRDF}\label{introduction:connection}
In Bayesian filtering {\cite{caines1988,elliott-aggoun-moore1995}}, one is given a model that generates the unobserved process $\{X_t:~t=0,\ldots,n\}$, via its conditional distribution $\{P_{X_t|X^{t-1}}(dx_t|x^{t-1}):t=0,\ldots,n\}$, or via discrete-time recursive dynamics, {and} a model that generates observed data obtained from sensors $\{Z_t:~t=0,\ldots,n\}$, via its conditional distribution $\{P_{Z_t|Z^{t-1},X^t}$ $(dz_t|z^{t-1},x^t):t=0,\ldots,n\}$, while an  estimate of the unobserved process {$\{X_t:~t=0,\ldots,n\}$, denoted by}
$\{\widehat{X}_t:~t=0,\ldots,n\}$, is constructed causally, based on the observed data $\{Z_t:~t=0,\ldots,n\}$. 
Thus, in Bayesian filtering theory, both models which generate the unobserved and observed processes, $\{X_t:~t=0,\ldots,n\}$ and $\{Z_t:~t=0,\ldots,n\}$, respectively, are given \emph{\'a priori}, while the estimator $\{\widehat{X}_t:~t=0,\ldots,n\}$ is a nonanticipative functional of the past information $Z^{t-1},~t=0,\ldots,n$, often computed recursively, like Kalman filter. Fig.~\ref{introduction:fig:filtering} illustrates the  block diagram of the Bayesian filtering problem.	
\begin{figure}[!h]
\centering
\includegraphics[width=0.79\columnwidth]{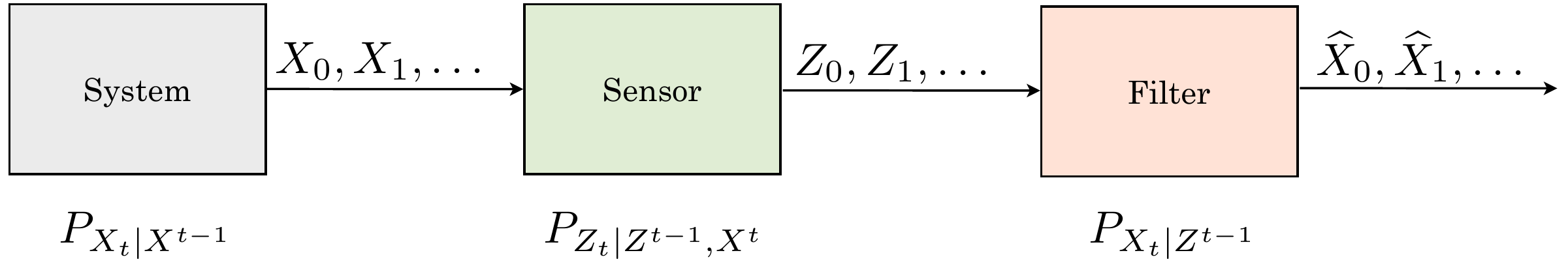}
\vspace{-0.1cm}
\caption{Bayesian Filtering Problem\vspace{-0.3cm}.}
\label{introduction:fig:filtering}
\end{figure}

\par On the other hand, in information-based estimation, defined in  Problem~\ref{fp1}, one is given the process  $\{X_t:~t=0,\ldots,n\}$ and a fidelity criterion,  and the objective is to  determine the optimal nonanticipative reproduction conditional distribution $\{P^*_{Y_t|Y^{t-1},X^t}(dy_t|y^{t-1},x^t):~t=0,\ldots,n\}$ corresponding to NRDF, denoted hereinafter by $R_{0,n}^{na}(D)$, and to realize this distribution by an \texttt{\{encoder, channel, decoder\}} so that the end-to-end distortion \eqref{fid_1} is met.

\par As a result, in Problem~\ref{fp1}, the observation model is constructed by  the cascade of the \texttt{\{encoder, channel\}} and the filter is the decoder, which satisfies the end-to-end fidelity \eqref{fid_1}.  

\subsection{Contributions}\label{introduction:contributions}
The main contributions of this paper are the following:\\
{\bf (R1)} We give a {\it closed form expression} for {the optimal nonanticipative reproduction conditional distribution,} $\{P^*_{Y_t|Y^{t-1}, X^t}: t=0, \ldots, n\}${, which} achieves the infimum of the Finite-Time Horizon (FTH) NRDF\footnote{In the sequel, when we refer to FTH NRDF we just say NRDF.}. {Then,} we identify some of its properties{, which are necessary for the design of the optimal \texttt{\{encoder, decoder\}} pair. 

\noi {\bf (R2)} We apply our framework to a {\it time-varying multidimensional fully observed Gauss-Markov process} $\{X_t: t=0, \ldots, n\}$ with MSE distortion,  and we show the following:
\begin{itemize}
\item[(1)] The parametric expression of $R_{0,n}^{na}(D)$ is characterized by a time-space reverse-waterfilling;   


\item[(2)] At each time $n$ the value  $R_{0,n}^{na}(D)$  is achieved by an optimal \texttt{\{encoder, channel, decoder\}}, where the channel is  a Multiple Input Multiple Output (MIMO) Additive Gaussian Noise (AGN) channel, the encoder operates at the capacity of the AGN channel, and \eqref{fid_1} holds with equality. 

\item[(3)] At each time $n$, we {give the universal lower bound on} {\it the  MSE of any causal estimator of the Gauss-Markov process}.
\end{itemize}


\par Contribution {\bf(R1)} generalizes \cite{charalambous-stavrou-ahmed2014ieeetac,stavrou-kourtellaris-charalambous2016ieeetit}, in that we remove the {assumption that} the optimal reproduction distribution $\{P^*_{Y_t|Y^{t-1}, X^t}: t=0, \ldots, n\}$ is time-invariant, the source is Markov, and  distortion is single-letter. This leads to recursive computation of the optimal nonstationary distribution $\{P^*_{Y_t|Y^{t-1}, X^t}: t=0, \ldots, n\}$,  backwards in time{; i.e.}, starting at time $t=n$ and going backward{s} to time $t=0$. Contribution {\bf(R2)} demonstrates that for time-varying multidimensional fully-observed Gauss-Markov processes, the parametric expression of the NRDF, $R_{0,n}^{na}(D)$, is characterized by a time-space reverse-waterfilling. To solve the time-space reverse-waterfilling,  we propose an iterative algorithm which  computes numerically the value of $R_{0,n}^{na}(D)$, and we present examples to illustrate the effectiveness of the algorithm. The Markovian property of the optimal reproduction distribution, implies that the optimal distribution  is  $\{P^*_{Y_t|Y^{t-1}, X_t}: t=0, \ldots, n\}$. This is realized by an \texttt{\{encoder, channel, decoder\}},  with probability of estimation error decaying exponentially{, under certain conditions}. The universal lower bound on the MSE of any estimator generalizes the well-known bound of a Gaussian RV given in  \cite{liptser-shiryaev2001}.   The new recursive estimator is finite-dimensional, and  ensures the fidelity constraint is met. The time-space reverse-waterfilling implies that given a distortion level, the optimal state estimation is chosen based on an optimal threshold policy, in time and space (dimension).  This is the fundamental difference from the well-known Kalman filter equations. 

\par The rest of the paper is structured as follows. {In Section~\ref{notation}, we provide the notation used throughout the paper.} In Section~\ref{nonanticipative_rdf}, we introduce NRDF for general processes. In Section~\ref{nonstationary:necessary_conditions_nonstationary_solution_rdf}, we describe the form of the optimal nonstationary (time-varying) reproduction distribution of the NRDF. In Section~\ref{nonstationary:section:example:gaussian}, we concentrate on evaluating the NRDF for time-varying multidimensional Gaussian processes with memory, present examples in the context of realizable filtering theory, and we derive a universal lower bound to the mean square error of any estimator of Gaussian processes based on NRDF. We draw conclusions and discuss future directions in Section~\ref{sec:conclusions}.

%
%
%
%
\section{Notation}\label{notation}

We let $\mathbb{R}=(-\infty,\infty)$, $\mathbb{Z}=\{\ldots,-1,0,1,\ldots\}$, {$\mathbb{N}=\{1,2, \ldots\}$,} $\mathbb{N}_0=\{0,1,\ldots\}$, $\mathbb{N}_0^n\triangleq\{0,1,\ldots,n\}$. $\mathbb{E}\{\cdot \}$ represents the expectation of its argument. $\sigma\{\cdot\}$ represents the $\sigma$-algebra of events generated by its argument. For a non-square matrix $A\in \mathbb{R}^{n\times m}$, we denote its transpose by $A\T$. 
 For a square matrix $A\in \mathbb{R}^{n\times n}$, we denote by $\diag\{A\}$ the matrix having $A_{ii},~i=1,\ldots,n$, on its diagonal and zero elsewhere. We denote the source alphabet spaces by the measurable space $\{({\cal X}_n,{\cal B}({\cal X }_n)):n\in\mathbb{Z}\}$, where ${\cal X}_n, n\in\mathbb{Z}$ are complete separable metric spaces or Polish spaces, and ${\cal B}({\cal X}_n)$ are Borel $\sigma-$algebras of subsets of ${\cal X}_n$. We denote points in ${\cal X}^{\mathbb{Z}}\triangleq{{\times}_{n\in\mathbb{Z}}}{\cal X}_n$ by $x_{-\infty}^{\infty}\triangleq\{\ldots,x_{-1},x_0,x_1,\ldots\}\in{\cal X}^{\mathbb{Z}}$, and their restrictions to finite coordinates for any $(m,n)\in\mathbb{N}_0$ by $x_{m}^n\triangleq\{x_{m},\ldots,x_0,x_1,\ldots,x_n\}\in{\cal X}_m^n,~n\geq{m}$.  We denote by ${\cal B}({\cal X}^{\mathbb{Z}})\triangleq\otimes_{t\in\mathbb{Z}}{\cal B}({\cal X}_t)$ the $\sigma-$algebra on ${\cal X}^{\mathbb{Z}}$ generated by cylinder sets $\{{\bf x}=(\ldots,x_{-1},x_0,x_1,\ldots)\in{\cal X}^{\mathbb{Z}}:x_j\in{A}_j,~j\in\mathbb{Z}\}, A_j\in{\cal B}({\cal X}_j), j\in\mathbb{Z}$. Thus, ${\cal B}({\cal X}_{m}^{n})$ denote the $\sigma-$algebras of cylinder sets in ${\cal X}_m^n$, with bases over $A_j\in{\cal B}({\cal X}_j)$,~$j\in\{m,m+1,\ldots,n\},~m\leq{n},~ (m,n)\in\mathbb{Z}$. For a Random Variable (RV) $X : (\Omega, {\cal F}) \longmapsto  ({\cal X}, {\cal B}({\cal X}))$  we denote the distribution induced by $X$ on $({\cal X}, {\cal B}({\cal X}))$ by ${\bf P}_{X}(dx)\equiv {\bf P}(dx)$. We denote the set of such probability distributions by $ {\cal M}({\cal X})$. We denote the conditional distribution of RV $Y$ given $X=x$ (i.e., fixed) by ${\bf P}_{Y|X}(dy| X=x)  \equiv {\bf P}_{Y|X}(dy|x)$. Such conditional distributions are  equivalently  described   by  stochastic kernels or transition functions \cite{dupuis-ellis1997} ${\bf K}(\cdot|\cdot)$ on $ {\cal  B}({\cal Y}) \times {\cal X}$, mapping ${\cal X}$ into ${\cal M}({\cal Y})$ (space of distributions), i.e., $x \in {\cal X}\longmapsto {\bf K}(\cdot|x)\in{\cal M}({\cal Y})$, and such that for every $A\in{\cal B}({\cal Y})$, the function ${\bf K}(A|\cdot)$ is ${\cal B}({\cal X})$-measurable. We denote the set of such stochastic kernels by ${\cal Q}({\cal Y}|{\cal X})$.
%
%
%
%
\section{NRDF on General Alphabets}\label{nonanticipative_rdf}

In this section, we introduce the definition of NRDF for general processes taking values in Polish spaces (complete separable metric spaces), that include finite, countable, and continuous alphabet spaces. 

\noi{\bf Source Distribution.} The process $\{X_0, X_1, \ldots\}$ is described by the  collection of conditional probability distributions $\{{\bf P}_{X_n|X^{n-1}}(\cdot|x^{n-1}): x^{n-1} \in {\cal X}^{n-1}, ~n\in\mathbb{N}_0\}$.  For each $n\in\mathbb{N}_0$, we let  ${\bf P}_{X_n|X^{n-1}}(\cdot|\cdot) \equiv P_n(\cdot| \cdot) \in{\cal Q}_n({\cal X}_n|{\cal X}^{n-1})$, and  for $n=0$, we set ${\bf P}_{X_0|X^{-1}}=P_0(dx_0)$. We define the probability distribution on ${\cal X}^{n}$ by
\begin{align}
{ P}_{0,n}(A_{0,n})\triangleq\int_{A_0}P_0(dx_0)\ldots\int_{A_n}P_n(dx_n|x^{n-1}),    \hso A_t\in{\cal B}({\cal X }_t), \; A_{0,n}=\times_{t=0}^n{A_t}.\label{equation2}
\end{align}
\noi Thus, for each $n\in\mathbb{N}_0$, ${P}_{0,n}(\cdot)\in{\cal M}({\cal X}^n)$.
\vspace*{0.2cm}\\
\noi{\bf Reproduction Distribution.} The  reproduction process $\{\ldots, Y_2, Y_1, Y_0, Y_1, \ldots\}\equiv \{Y^{-1}, Y_0, Y_1, \ldots, \}$  is described by the collection of conditional  distributions $\{{\bf P}_{Y_n|Y^{n-1},X^n}$ $(\cdot|y^{n-1},x^n): (y^{n-1}, x^n) \in  {\cal Y}^{n-1} \times {\cal X}^n,  n\in\mathbb{N}_0\}$, i.e., $y^n\equiv ({y^{-1}, y_0^n}), x^n\equiv x_0^n$. For each $n\in\mathbb{N}_0$, we let ${\bf P}_{Y_n|Y^{n-1},X^n}(\cdot|\cdot,\cdot)\equiv Q_n(\cdot|\cdot,\cdot)\in{\cal Q}_n({\cal Y}_n|{\cal Y}^{n-1}\times{\cal X}^{n})$, and for $n=0$,  
${\bf P}_{Y_0|Y^{-1},X_0}=Q_0(dy_0|y^{-1},x_0)$. The RV $Y^{-1}$ is the initial data with fixed distribution ${\bf P}_{Y^{-1}}(dy^{-1})=\mu(dy^{-1})$. 
We define the family of conditional probability distributions on ${\cal Y}_0^{n}$ parametrized by $(y^{-1},x^n) \in {\cal Y}^{-1} \times {\cal X}^{n}$ by 
\begin{align}
{\overrightarrow{Q}}_{0,n}(B_{0,n}|y^{-1},x^n)\triangleq & \int_{{B}_0}Q_0(dy_0|y^{-1},x_0)\ldots\nonumber \\
&\int_{{B}_n}Q_n(dy_n|y^{n-1},x^n), \hso B_t\in{\cal B}({\cal Y}_t),\: {B}_{0,n}=\times_{t=0}^n{{B}_t}.\label{nonstationary:equation4}
\end{align}
\noi Thus, for each $n\in\mathbb{N}_0$, $\overrightarrow{Q}_{0,n}(\cdot|y^{-1},x^n)\in{\cal M}({\cal Y}_0^{n}),~(y^{-1},x^{n})\in {\cal Y}^{-1} \times {\cal X}^{n}$. 

\noi Given  a ${P}_{0,n}(\cdot)\in{\cal M}({\cal X}^{n})$ a $\overrightarrow{Q}_{0,n}(\cdot|y^{-1},x^n)\in{\cal M}({\cal Y}_0^{n})$, and a fixed distribution $\mu(dy^{-1})$, we define the following distributions.\\
{\it The joint distribution} on ${\cal X}^{n}\times{\cal Y}_0^{n}$ given  $Y^{-1}=y^{-1}$ is defined by  
\begin{align}
{\bf P}^{\overrightarrow{Q}}(A_{0,n} \times B_{0,n}|y^{-1})\triangleq &({P}_{0,n}{\otimes}{\overrightarrow Q}_{0,n})\left(\times^n_{t=0}(A_t{\times}B_t)|y^{-1}\right)\nonumber\\
=&\int_{A_0}P_0(dx_0) \int_{{B}_0}Q_0(dy_0|y^{-1},x_0)  \ldots \nonumber \\
&\int_{A_n}P_n(dx_n|x^{n-1})\int_{{B}_n}Q_n(dy_n|y^{n-1},x^n). \label{joint:distribution}
\end{align}
{\it The marginal distribution} on ${\cal Y}_0^{n}$  given $Y^{-1}=y^{-1}$ is defined by 
\begin{align*}
{\Pi }_{0,n}^{\overrightarrow{Q}}(B_{0,n}|y^{-1}) \triangleq & \int_{B_{0,n}} \int_{{\cal X}^n}  ({P}_{0,n}\otimes{\overrightarrow Q}_{0,n})(dx^n, dy_0^n|y^{-1})\\
=& \int_{B_{0,n}} \Pi_0^{\overrightarrow{Q}}(dy_0|y^{-1}) \ldots \Pi_n^{\overrightarrow{Q}}(dy_n|y^{n-1}).
\end{align*}
 {\it The product probability distribution} ${\overrightarrow\Pi}_{0,n}^{\overrightarrow{Q}}(\cdot|y^{-1}):{\cal B}({\cal X}^n)\otimes{\cal B}({\cal Y}_0^{n})\longmapsto[0,1]$ conditioned on $Y^{-1}=y^{-1}$, is defined by
\begin{align}
&{\overrightarrow\Pi}_{0,n}^{\overrightarrow{Q}}\left(A_{0,n} \times B_{0,n}|y^{-1}\right)\triangleq \left({P}_{0,n}\times{\Pi}_{0,n}^{\overrightarrow{Q}}\right)\left(\times^n_{t=0}(A_t{\times}{B_t})|y^{-1}\right)\nonumber\\
&=\int_{A_0}P_{0}(dx_0)\int_{B_0}\Pi_0^{\overrightarrow{Q}}(dy_0|y^{-1})\ldots\int_{A_n}P_n(dx_n|x^{n-1})\int_{B_n}\Pi_n^{\overrightarrow{Q}}(dy_n|y^{n-1}).\nonumber
\end{align}

\par We define the relative entropy between the joint distribution ${\bf P}^{\overrightarrow{Q}}(dx^n,dy_0^n|y^{-1})$ and the product distribution  ${\overrightarrow\Pi}_{0,n}^{\overrightarrow{Q}}\left(dx^n,dy_0^n|y^{-1}\right)$, averaged over the initial distribution $\mu(dy^{-1})$, as follows:
\begin{align}
\mathbb{D}({P}_{0,n} \otimes {\overrightarrow Q}_{0,n}||{\overrightarrow\Pi}_{0,n}^{\overrightarrow{Q}})= &\int_{{\cal X}^{n}\times{\cal Y}^{n}} \log \left( \frac{P_{0,n}(\cdot) \otimes {\overrightarrow Q}_{0,n}(\cdot|y^{-1},x^n)}{P_{0,n}(\cdot)\otimes {\Pi}_{0,n}^{\overrightarrow{Q}}\left(\cdot|y^{-1}\right)      }(x^n, y_0^n) \right) \nonumber \\
& P_{0,n}(dx^n) \otimes {\overrightarrow Q}_{0,n}(dy_0^n|y^{-1},x^n)\otimes \mu(dy^{-1}) 
\label{equation33}  \\
\sr{(a)}{=}& \int_{{\cal X}^n\times{\cal Y}^{n}} \log \left( \frac{{\overrightarrow Q}_{0,n}(\cdot|y^{-1},x^n)}{{\Pi }_{0,n}^{\overrightarrow{Q}}(\cdot|y^{-1}) }(y_0^n)\right)\nonumber \\
& P_{0,n}(dx^n) \otimes {\overrightarrow Q}_{0,n}(dy_0^n|y^{-1},x^n)\otimes \mu(dy^{-1}) \label{equation203}\\
\equiv & {\mathbb{I}}_{0,n}({P}_{0,n}, {\overrightarrow Q}_{0,n})\label{equation7a}
\end{align}
where $(a)$ is due to the chain rule of relative entropy (see \cite{charalambous-stavrou2016ieeetit}). 
In (\ref{equation7a})  the notation ${\mathbb{I}}_{0,n}(\cdot,\cdot)$  indicates the functional dependence  on $\{P_{0,n}, {\overrightarrow Q}_{0,n}\}$ (the dependence on $\mu(dy^{-1})$ is omitted). By  \cite[Theorem 5]{charalambous-stavrou2016ieeetit},  the set of distributions $\overrightarrow{Q}_{0,n}(\cdot|y^{-1},x^n)\in{\cal M}({\cal Y}_0^{n})$ is convex, and by \cite[Theorem 6]{charalambous-stavrou2016ieeetit}, ${\mathbb{I}}_{0,n}(P_{0,n},\cdot)$ is a convex functional of $\overrightarrow{Q}_{0,n}(\cdot|y^{-1},x^n)\in{\cal M}({\cal Y}_{0}^n)$. 

\par Given the distortion function  of reproducing $x_t$ by $y_t, t=0,1, \ldots, n$, defined  by \eqref{introduction:distortion_function}, the  fidelity constraint set  is defined as follows.
\begin{align}
\overrightarrow{\cal Q}_{0,n}&(D)\triangleq\left\{\overrightarrow{Q}_{0,n}(\cdot|y^{-1},x^n)\in{\cal M}({\cal Y}_0^{n}):\frac{1}{n+1}{\bf E}_{\mu}^{\overrightarrow{Q}}\left\{d_{0,n}(X^n,Y^n)\right\}\leq{D}\right\},~D\geq{0} \nonumber
\end{align}
where ${\bf E}_{\mu}^{\overrightarrow{Q}}\{\cdot\}$ indicates that the joint distribution is induced by $\{P_{0,n}(dx^n), \overrightarrow{Q}_{0,n}(dy_0^n|y^{-1}$,  $ x^n), {\mu}(dy^{-1})\}$ defined by (\ref{joint:distribution}). Clearly, $\overrightarrow{\cal Q}_{0,n}(D)$ is a convex set. 

\begin{definition}({NRDF}){\ \\}
The NRDF is defined by 
\begin{align}
{R}^{na}_{0,n}(D) \triangleq  \inf_{{\overrightarrow{Q}_{0,n}(dy_0^n|y^{-1}, x^n)\in\overrightarrow{\cal Q}_{0,n}(D)}}\mathbb{I}_{0,n}(P_{0,n},{\overrightarrow Q}_{0,n}).
\label{nonstationary:ex12}
\end{align}
\end{definition}
By the above discussion the NRDF is a convex optimization problem. Sufficient conditions for existence of an optimal solution to the convex optimization problem \eqref{nonstationary:ex12} are given in \cite[Theorem III.13]{charalambous-stavrou2016ieeetit}.

For completeness, in the next remark we give the connection of the NRDF to the classical Shannon RDF \cite{berger1971} and  nonanticipatory $\epsilon$-entropy \cite{gorbunov-pinsker1973}.

\begin{remark}(RDF and nonanticipatory $\epsilon$-entropy) {\ \\}
Consider the distribution ${P}_{0,n}(\cdot)\in{\cal M}({\cal X}^n)$ and the conditional distribution ${Q}^{\nc}_{0,n}(dy_0^n|$ $y^{-1},x^n)\in{\cal M}({\cal Y}_0^{n}),~(y^{-1},x^{n})\in {\cal Y}^{-1} \times {\cal X}^{n}$, which is a non-causal distribution, because by Bayes' rule ${Q}^{\nc}_{0,n}(dy_0^n|y^{-1},x^n)=\otimes_{t=0}^n{Q}_t^{\nc}(dy_t|y^{t-1},x^n)$. The conditional  distribution on ${\cal Y}_0^n$ given $Y^{-1}=y^{-1}$, and the joint distribution on ${\cal X}^n \times {\cal Y}_0^n$ are  induced as follows.   
\begin{align}
\Pi_{0,n}^{Q^{\nc}}(dy_0^n|y^{-1})=&\int_{{\cal X}^n}Q^{\nc}_{0,n}(dy_0^n|y^{-1}, x^n)\otimes P_{0,n}(dx^n), \\
 {\bf P}^{Q^{\nc}}(dx^n, dy_0^n|y^{-1})= & P_{0,n}(dx^n) \otimes Q^{\nc}_{0,n}(dy_0^n|y^{-1},x^n).
\end{align}
Define the fidelity constraint 
\begin{align}
{\cal Q}_{0,n}(D)\triangleq\left\{ {Q}^{\nc}_{0,n}(dy_0^n|y^{-1},x^n)\in{\cal M}({\cal Y}_0^n):\frac{1}{n+1}{\bf E}_{\mu}^{Q^{\nc}}\left\{d_{0,n}(X^n,Y^n)\right\}\leq{D}\right\},~D\geq{0}.\label{fidelity_set_general}
\end{align} 
The classical RDF \cite{berger1971} is defined by 
\begin{align}
{R}_{0,n}(D) \triangleq  \inf_{{{Q}^{\nc}_{0,n}(dy_0^n|y^{-1}, x^n)\in {\cal Q}_{0,n}(D)}} I(X^n;Y_0^n|Y^{-1}) ,
\label{nonstationary:ex12_c}
\end{align}
where $I(X^n;Y_0^n|Y^{-1})$ is the conditional mutual information given by 
\begin{align}
I(X^n; Y_0^n|Y^{-1}) =&\int_{{\cal X}^n\times{\cal Y}^n} \log \left( \frac{{ Q}^{\nc}_{0,n}(\cdot|y^{-1},x^n)}{{\Pi}_{0,n}^{Q^{\nc}}\left(\cdot|y^{-1}\right)      }(y_0^n) \right) \nonumber \\
& P_{0,n}(dx^n) \otimes { Q}^{\nc}_{0,n}(dy_0^n|y^{-1},x^n)\otimes \mu(dy^{-1}) 
\label{equation33_N}  \\
\equiv & \mathbb{I}_{0,n}(P_{0,n}, Q^{\nc}_{0,n}).\label{equation7a_N}
\end{align}
Unfortunately, classical RDF does not give causal estimators, because the optimal reproduction distribution   in \eqref{nonstationary:ex12_c} is $\{{Q}^{\nc}_{t}(dy_t|y^{t-1}, x^n): t\in\mathbb{N}_0^n\}$; hence, in general, it is non-causal with respect to $\{X_0, \ldots, X_n\}$. This let Gorbunov and Pinsker  in \cite{gorbunov-pinsker1973} to define the notion of nonanticipatory $\epsilon$-entropy,  as follows
\begin{align}
R_{0,n}^{\varepsilon}(D)\triangleq\inf_{\substack{{\cal Q}_{0,n}(D): Q^{\nc}_{0,t}(dy_0^t|y^{-1}, x^n)=Q_{0,t}^{GP}(dy_0^t|y^{-1}, x^t)  \\ t=0, \ldots, n}} I(X^n;Y_0^n|Y^{-1}). \label{introduction:definition:epsilon-entropy}
\end{align}
We note that conditional independence $Q^{\nc}_{0,t}(dy_0^t|y^{-1}, x^n)=Q_{0,t}^{GP}(dy_0^t|y^{-1}, x^t), t=0, \ldots, n$ is a causality restriction of the reproduction distribution in (\ref{nonstationary:ex12_c}).

\par The equivalence of the nonanticipatory $\epsilon$-entropy, $R_{0,n}^{\varepsilon}(D)$, and NRDF, $R_{0,n}^{na}(D)$, is a direct consequence of the following equivalent characterization of conditional independence  statements shown in \cite{stavrou-kourtellaris-charalambous2016ieeetit}.

\begin{description}
\item[{ MC1}:] ${Q}^{\nc}_{0,n}(dy_0^n|y^{-1},x^n)={\overrightarrow Q}_{0,n}(dy_0^n|y^{-1},x^n)=\otimes_{t=0}^n{Q}_{t}(dy_t|y^{t-1},x^t)$,~~$\forall{n}\in\mathbb{N}_0$;

\item[{ MC2}:]  $Q^{\nc}_t(dy_t|y^{t-1},x^t,x_{t+1}^n)=Q_t(dy_t|y^{t-1},x^t)$,~for each $t=0,1,\ldots, n-1$,~$\forall{n}\in\mathbb{N}_0$;

\item[{ MC3}:] $P_t(dx_{t+1}|x^{t},y^t)=P_t(dx_{t+1}|x^{t})$,~for each  $t=0,1,\ldots, n-1$,~$\forall{n}\in\mathbb{N}_0$;

\item[{ MC4}:] $Q^{\nc}_{0,t}(dy_0^t|y^{-1},x^{t},x_{t+1}^n)=\overrightarrow{Q}_{0,t}(dy_0^t|y^{-1},x^{t})$,~for each $t=0,1,\ldots, n-1$,~$\forall{n}\in\mathbb{N}_0$. 
\end{description}

\noi In view of the above statements, the  NRDF defined by (\ref{nonstationary:ex12}) is equivalent to the nonanticipatory $\epsilon$-entropy defined by (\ref{introduction:definition:epsilon-entropy}), that is, $R_{0,n}^{na}(D)= R_{0,n}^\varepsilon(D)$.
\end{remark}

%
%
%
%
\section{Optimal Nonstationary Reproduction Distribution}\label{nonstationary:necessary_conditions_nonstationary_solution_rdf}

In this section, we describe the form of the optimal nonstationary (time-varying) reproduction distribution that achieves the infimum in (\ref{nonstationary:ex12}). 

\par First, we state the following properties regarding the convexity and continuity of the NRDF, ${R}^{na}_{0,n}(D)$, that are necessary for the development of our results.\\
{\bf 1)}~${R}^{na}_{0,n}(D)$ is a convex, non-increasing function of $D\in[0,\infty)$.\\
{\bf 2)}~If $R_{0,n}^{na}(D)<\infty$, then $R_{0,n}^{na}(\cdot)$ is continuous on $D\in[0,\infty)$.\\
Note that {\bf1)} is similar to the one derived in \cite[Lemma IV.4]{stavrou-kourtellaris-charalambous2016ieeetit}. Also, for {\bf2)} recall that a bounded and convex function is continuous. Since $R_{0,n}^{na}(D)$ is non-increasing, it is bounded outside the neighbourhood of $D=0$ and it is also continuous on $(0,\infty)$. In other words, if $R_{0,n}^{na}(D)<\infty$ then $R_{0,n}^{na}(D)$ is bounded and hence continuous on $[0,\infty)$. \\
Moreover, since $R_{0,n}^{na}(D)$ is convex and non-increasing then its inverse function, $D_{0,n}(R^{na})$, exists and it is convex, non-increasing function of $R^{na}\in[0,\infty)$. $D_{0,n}(R^{na})$ is called FTH Nonanticipative Distortion Rate Function (NDRF) and is given by
\begin{align}
D_{0,n}(R^{na})=\inf_{\frac{1}{n+1}{\mathbb I}_{0,n}(P_{0,n}, \overrightarrow{Q}_{0,n})\leq{R^{na}}}{\bf E}_{\mu}^{\overrightarrow{Q}}\left\{ d_{0,n}(X^n, Y^n)\right\}.\label{section:ndrf:equation1}
\end{align}

\noi The NRDF defined by (\ref{nonstationary:ex12}) is a convex optimization problem, and thus, if there exists an interior point in the set $\overrightarrow{\cal Q}_{0,n}(D)$, it can be reformulated using Lagrange duality theorem \cite[Theorem 1, pp. 224-225]{dluenberger1969} as an unconstrained problem as follows.
\begin{align}
R^{na}_{0,n}(D)= \sup_{s \leq 0} \inf_{\overrightarrow{Q}_{0,n}(\cdot|y^{-1},x^n) \in{\cal M}({\cal Y}_0^{n})}\left\{\mathbb{I}_{0,n}({P}_{0,n},\overrightarrow{Q}_{0,n})  - s\frac{1}{n+1}{\bf E}_{\mu}^{\overrightarrow{Q}}\left\{d_{0,n}(X^n,Y^n)\right\}\right\}. \label{nonstationary:eq.0}
\end{align}

\par Next, we state Theorem~\ref{nonstationary:theorem:solution_rdf}, which is used in the subsequent analysis to compute the NRDF, $R^{na}_{0,n}(D)$, of time-varying multidimensional Gauss-Markov processes.

\begin{theorem}(Optimal nonstationary reproduction distributions)\label{nonstationary:theorem:solution_rdf}{\ \\}
Suppose there exists a $\overrightarrow{Q}^*_{0,n}(\cdot|y^{-1},x^n)\in\overrightarrow{\cal Q}_{0,n}(D)$, which solves (\ref{nonstationary:ex12}), and that $\mathbb{I}_{0,n}(P_{0,n}$, $\overrightarrow{Q}_{0,n})$ is G\^ateaux differentiable in every direction of $\{Q_{t}(\cdot|y^{t-1},x^t): {t\in\mathbb{N}_0^n}\}$ for a fixed $P_{0,n}(\cdot)\in{\cal M}({\cal X}^n)$ and $\mu(dy^{-1})\in{\cal M}({\cal Y}^{-1})$. Then, the following hold: \\
{\bf(1)} The optimal reproduction distributions denoted by  $\{Q^*_{t}(\cdot|y^{t-1},x^t)\in{\cal M}({\cal Y}_t): {t\in\mathbb{N}_0^n}\}$ are given by the following recursive equations backwards in time.\\
For $t=n$:
\begin{align}
&Q^*_{n}(dy_n|y^{n-1},x^n)=\frac{{e}^{s\rho_n(T^n{x^n},{T^ny^n})}\Pi^{\overrightarrow{Q}^*}_{n}(dy_n|y^{n-1})}{\int_{{\cal Y}_n} e^{s \rho_n(T^n{x^n},T^n{y^n})}\Pi^{\overrightarrow{Q}^*}_{n}(dy_n|y^{n-1})}.\label{nonstationary:eq.31a}
\end{align}
For $t=n-1,n-2,\ldots,0$:
\begin{align}
&Q^*_{t}(dy_t|y^{t-1},x^t)=\frac{{e}^{s\rho_t(T^t{x^n},T^t{y^n})-g_{t,n}(x^t,y^t)}\Pi^{\overrightarrow{Q}^*}_{t}(dy_t|y^{t-1})}{\int_{{\cal Y}_t} e^{s \rho_t(T^t{x^n},T^t{y^n})-g_{t,n}(x^t,y^t)}\Pi^{\overrightarrow{Q}^*}_{t}(dy_t|y^{t-1})}\label{nonstationary:eq.31b}
\end{align}
where $s<0$, $\Pi^{\overrightarrow{Q}^*}_{t}(\cdot|y^{t-1})\in {\cal M}({\cal Y}_t)$ and $g_{t,n}(x^t,y^t)$ is given by
\begin{align}\label{nonstationary:eq.32b}
\begin{split}
&g_{t,n}(x^{t},y^{t})=-\int_{{\cal X}_{t+1}}{P}_{{t+1}}(dx_{t+1}|x^t)\log\left(\int_{{\cal Y}_{t+1}}{e}^{s\rho_{t+1}(T^{t+1}x^{n},T^{t+1}y^{n})-g_{t+1,n}(x^{t+1},y^{t+1})}\Pi^{\overrightarrow{Q}^*}_{t+1}(dy_{t+1}|y^t)\right),\\
& g_{n,n}(x^n, y^n)=0.
\end{split}
\end{align} 
{\bf(2)} The {NRDF} is given by
\begin{align}
&{R}^{na}_{0,n}(D)=sD-\frac{1}{n+1}\sum_{t=0}^n\int_{{\cal X}^t\times{\cal Y}^{t-1}}\bigg\{\int_{{\cal Y}_t}{g}_{t,n}(x^t,y^t)Q^*_{t}(dy_t|y^{t-1},x^t)\nonumber\\
&\qquad+\frac{1}{n+1}\log\left(\int_{{\cal Y}_t}{e}^{s\rho_t(T^tx^n,T^ty^n)-g_{t,n}(x^t,y^t)}\Pi^{\overrightarrow{Q}^*}_{t}(dy_{t}|y^{t-1})\right)\bigg\} \nonumber \\
&\qquad\qquad\otimes{P}_{t}(dx_t|x^{t-1})\otimes({P}_{0,t-1}\otimes\overrightarrow{Q}_{0,t-1}^*)(dx^{t-1},dy_0^{t-1}|y^{-1})\otimes{\mu}(dy^{-1}).\label{nonstationary:eq.33}
\end{align}
{\bf (3)} If ${R}^{na}_{0,n}(D) > 0$ then $ s < 0$,  and
\begin{equation}
\frac{1}{n+1}\sum_{t=0}^n\int_{{\cal X}^t\times{\cal Y}^t}\rho_t(T^tx^n,T^ty^n)({P}_{0,t}\otimes\overrightarrow{Q}_{0,t}^*)(dx^{t},dy_0^{t}|y^{-1})\otimes{\mu}(dy^{-1})=D\label{nonstationary:eq.33i}.
\end{equation}
\end{theorem}
\begin{proof}
The sequence of minimizations over $\{Q_{t}(\cdot|y^{t-1},x^t):~{t\in\mathbb{N}_0^n}\}$ in \eqref{nonstationary:eq.0} is a nested optimization problem. Hence, we can introduce the dynamic programming recursive equations. Then, we carry out the infimum starting at the last stage over $Q_{n}(\cdot|y^{n-1},x^n)\in{\cal M}({\cal Y}_n)$ and sequentially move backwards in time to determine $Q^*_{n}(\cdot|y^{n-1},x^n), Q^*_{n-1}(\cdot|y^{n-2},x^{n-1}),\ldots, Q^*_{0}(\cdot|y^{-1},x_0)$. The procedure is straightforward and we omit it due to space limitations. 
\end{proof}
\par We note that Theorem \ref{nonstationary:theorem:solution_rdf} is fundamentally different from \cite[Theorem IV.4]{charalambous-stavrou-ahmed2014ieeetac}. In the latter, it is assumed that all elements $\{Q_t(dy_t|y^{t-1},x^t): t\in\mathbb{N}_0^n\}$ are identical.

\par From  the above theorem, for a given distribution $P_{0,n}(\cdot)\in{\cal M}({\cal X}^n)$, we can identify the dependence of the optimal nonstationary reproduction distribution on past and present symbols of the information process $\{X_t:~t\in\mathbb{N}_0^n\}$, but not its dependence on past reproduction symbols. In what follows, we give certain properties of the information structure of the optimal nonstationary reproduction distribution that achieves the infimum in (\ref{nonstationary:ex12}).
\vspace*{0.2cm}\\
\noi {\bf Information structure of the optimal nonstationary reproduction distribution}. \\
{\bf(1)}~The dependence of ${Q}^*_{n}(dy_n|y^{n-1},x^n)$ on $x^n\in{\cal X}^n$ is determined by the dependence of $\rho_n(T^nx^n,T^ny^n)$ on $x^n\in{\cal X}^n$ as follows:\\
{\bf (1.1)}  If $\rho_t(T^tx^n,T^iy^n)=\bar{\rho}(x_t,y^t), t=0, \ldots, n$, then ${Q}^*_{n}(dy_n|y^{n-1},x^n)={Q}^*_{n}(dy_n|y^{n-1},x_n)$, while for $t=n-1,n-2,\ldots,0$, the dependence of ${Q}^*_{t}(dy_t|y^{t-1},x^t)$ on $x^t\in{\cal X}^t$ is determined from the dependence of  $g_{t,n}(x^t,y^t)$ on $x^t\in{\cal X}^t$. \\
{\bf (1.2)} If $P_{t}(dx_{t}|x^{t-1})=P_{t}(dx_{t}|x_{t-1-L}^{t-1})$, where $L$ is a non-negative finite integer, and $\rho_t(T^tx^n,T^ty^n)=\bar{\rho}(x_{t-N}^t,y_t)$, where $N$ is a non-negative finite integer, then ${Q}^*_{t}(dy_t|$ $y^{t-1},x^t)={Q}^*_{t}(dy_t|y^{t-1},x_{t-J}^t)$, where $J=\max\{N,L\}$.\\
{\bf(2)} If $g_{t,n}(x^t,y^t)=\hat{g}_{t,n}(x^t,y^{t-1}),~\forall{t\in\mathbb{N}_0^{n-1}}$  then the optimal reproduction distribution (\ref{nonstationary:eq.31b}) reduces to
\begin{align}
{Q}^*_{t}(dy_t|y^{t-1},x^t)=\frac{{e}^{s\rho_t(T^t{x^n},T^t{y^n})}\Pi^{\overrightarrow{Q}^*}_{t}(dy_t|y^{t-1})}{\int_{{\cal Y}_t} e^{s\rho_t(T^t{x^n},T^t{y^n})}\Pi^{\overrightarrow{Q}^*}_{t}(dy_t|y^{t-1})}.\nonumber
\end{align}

\par To further understand the dependence of the optimal nonstationary reproduction distributions (\ref{nonstationary:eq.31a}), (\ref{nonstationary:eq.31b}) on past reproductions, we  state an alternative characterization of the nonstationary solution of $R^{na}_{0,n}(D)$, as a maximization over a certain class of functions. We use this additional characterization to derive lower bounds on $R^{na}_{0,n}(D)$, which are achievable.

\begin{theorem}(Characterization of solution of NRDF)\label{nonstationary:theorem:alternative_expression_nonstationary} {\ \\}
An alternative characterization of NRDF is
\begin{align}
\begin{split}
&R^{na}_{0,n}(D)=\sup_{s\leq{0}}\sup_{\{\lambda_t\in\Psi_s^t:~t\in\mathbb{N}_0^n\}}\Bigg\{sD-\frac{1}{n+1}\sum_{t=0}^n\int_{{\cal X}^t\times{\cal Y}^{t-1}} \int_{{\cal Y}_t}g_{t,n}(x^t,y^t)Q^*_{t}(dy_t|y^{t-1},x^t) \\
&\qquad+\log\left(\lambda_t(x^t,y^{t-1})\right) P_{t}(dx_t|x^{t-1})\otimes({P}_{0,t-1}\otimes\overrightarrow{Q}_{0,t-1}^*)(dx^{t-1},dy_0^{t-1}|y^{-1})\otimes\mu(dy^{-1})\Bigg\}, 
\end{split}\label{nonstationary:eq.38a}
\end{align}
where 
\begin{align}
\begin{split}
&\Psi_s^t\triangleq\bigg\{\lambda_t(x^t,y^{t-1})\geq{0}:~\int_{{\cal X}^{t-1}}\left(\int_{{\cal X}_{t}}e^{s\rho_t(T^tx^n,T^ty^n)-g_{t,n}(x^t,y^t)}\lambda_t(x^t,y^{t-1})P_{t}(dx_t|x^{t-1})\right) \\
&\qquad\qquad\otimes{\bf P}^{\overrightarrow{Q}^*}(dx^{t-1}|y^{t-1})\leq{1}\Bigg\}
\end{split}\label{nonstationary:eq.39}
\end{align}
and $g_{n,n}(x^n,y^n)=0$,~ and for $t\in\mathbb{N}_0^{n-1}$,
\begin{align}
g_{t,n}(x^{t},y^{t})=-\int_{{\cal X}_{t+1}}{P}_{{t+1}}(dx_{t+1}|x^t)\log\left(\lambda_{t+1}(x^{t+1},y^{t})\right)^{-1}.\nonumber
\end{align}
For $s\in(-\infty,0]$ a necessary and sufficient condition for $\{\lambda_t(\cdot,\cdot):~t=0,\ldots,n\}$ to achieve the supremum of (\ref{nonstationary:eq.38a}) is the existence of a probability distribution $\Pi_t^{\overrightarrow{Q}^*}(\cdot|y^{t-1})$ $\in{\cal M}({\cal Y}_t)$ such that
\begin{align}
\lambda_t(x^t,y^{t-1})=\left(\int_{{\cal Y}_t}e^{s\rho_t(T^tx^n,T^ty^n)-g_{t,n}(x^t,y^t)}\Pi^{\overrightarrow{Q}^*}_{t}(dy_t|y^{t-1})\right)^{-1},~t\in\mathbb{N}_0^n.\nonumber
\end{align} 
\end{theorem}
\begin{proof} See Appendix~\ref{nonstationary:appendix:proof:theorem:alternative_expression_nrdf}.
\end{proof}

\par Theorem~\ref{nonstationary:theorem:alternative_expression_nonstationary} is crucial in the computation of $R_{0,n}^{na}(D)$ for any given source (with memory), simply because apart from Gaussian or memoryless sources, to solve a rate distortion problem explicitly, one needs to identify the dependence of the optimal reproduction distribution on past reproduction symbols, $Y^{t-1}$, and in general to find the information structure of the optimal reproduction distribution. In the next section, we use the previous theorems to derive $R_{0,n}^{na}(D)$ for the Gaussian source.

%
%
%
%
\section{NRDF of Time-Varying Multidimensional Gauss-Markov Processes}\label{nonstationary:section:example:gaussian}

\par In this section, we apply Theorem~\ref{nonstationary:theorem:solution_rdf} and Theorem~\ref{nonstationary:theorem:alternative_expression_nonstationary} from Section~\ref{nonstationary:necessary_conditions_nonstationary_solution_rdf} to time-varying multidimensional Gauss-Markov processes in state-space form, and we obtain the following results:\\
\noi{\bf (1)} the analytical expression of the optimal nonstationary reproduction distribution that achieves the infimum of the NRDF and the analytical expression of the NRDF subject to a square error distortion;\\
\noi{\bf (2)} a realization of the optimal nonstationary reproduction distribution in the sense of Fig.~\ref{nonstationary:communication_system} that allows us to obtain the optimal filter;\\
\noi{\bf (3)} a universal lower bound on the MSE of any causal estimator of Gaussian processes.

\par The analytical expression of the  NRDF is found by developing a time-space algorithm, which is a generalization of the standard reverse-waterfilling algorithm derived in \cite[Section 10.3.3]{cover-thomas2006} for independent Gaussian RV. Toward this, illustrative examples that verify our theory are presented.



\par The time-varying multidimensional Gauss-Markov processes defined as follows.

\begin{definition}(Time-varying multidimensional Gauss-Markov process){\ \\}
The source process is modeled as a time-varying $p$-dimensional Gauss-Markov process  defined by 
\begin{align}
X_{t+1}=A_tX_t+B_tW_t, \:X_0=x_0,  \: \: {t\in\mathbb{N}_0^{n-1}},  \label{nonstationary:example:gaussian:equation51}
\end{align}
where $A_t\in\mathbb{R}^{p\times{p}}, B_t\in\mathbb{R}^{p\times{k}}, {t\in\mathbb{N}_0^{n-1}}$. We assume\\
\noi {\bf(G1)} $X_0\in\mathbb{R}^p$ is Gaussian $N(0;\Sigma_{X_0})$;\\
\noi {\bf(G2)} $\{W_t: {t\in\mathbb{N}_0^{n}}\}$ is a $k$-dimensional  $\IID$ Gaussian $N(0; I_{k})$ sequence, independent of $X_0$;\\
\noi {\bf (G3)} The distortion function is defined by $d_{0,n}(x^n,{y}^n)\triangleq\sum_{t=0}^n\rho_t(T^tx^n,T^ty^n)$ $=\sum_{t=0}^n||x_t-{y}_t||_{2}^2$.
\end{definition}

\paragraph*{\it Information Structure} By Theorem~\ref{nonstationary:theorem:solution_rdf} and the Markovian property of (\ref{nonstationary:example:gaussian:equation51}), the optimal nonstationary reproduction distribution given by \eqref{nonstationary:eq.31a}-\eqref{nonstationary:eq.31b} is Markov with respect to $\{X_0,\ldots,X_n\}$, that is, $\{{Q}^*_{t}(d{y}_t|y^{t-1},x^t)\equiv{Q}^*_{t}(d{y}_t|y^{t-1},x_t):~t\in\mathbb{N}_0^n\}$ (see the comments below Theorem~\ref{nonstationary:theorem:solution_rdf} on information structures of the optimal reproduction distribution). Since $\{X_t: t\in\mathbb{N}_0^n\}$ is Markov and the distortion function is squared error, then by \cite{tatikonda-sahai-mitter2004ieeetac} the optimal reproduction process $\{Y^*_t: t\in\mathbb{N}_0^n\}$ is Gaussian, and the joint process $\{(X_t,Y_t): t\in\mathbb{N}_0^n\}$ is also Gaussian. In what follows, we also show the Gaussianity of the structure of the optimal reproduction distribution $\{{Q}^*_{t}(d{y}_t|y^{t-1},x_t):~t\in\mathbb{N}_0^n\}$.

\par Starting from stage $n$ and going backwards, we can show that $\{{Q}^*_{t}$ $(d{y}_t|y^{t-1},x_t):~t\in\mathbb{N}_0^n\}$ are conditional Gaussian distributions.   
\paragraph*{Stage $n$} Since the exponential term $||{y}_n-x_n||_{2}^2$ in the Right-Hand Side (RHS) of (\ref{nonstationary:eq.31a}) is quadratic in $(x_n,{y}_n)$, and $\{X_t:~{t\in\mathbb{N}_0^{n}}\}$ is Gaussian,  then it follows that a Gaussian distribution $Q_{n}(\cdot|{y}^{n-1},x_n)$, for a fixed realization of $({y}^{n-1},x_n)$, and {a} Gaussian distribution $\Pi^{\overrightarrow{Q}}_{n}(\cdot|{y}^{n-1})$ satisfy both the left and right sides of (\ref{nonstationary:eq.31a}). This implies that $Q^*_{n}(\cdot|y^{n-1},x_n)$ and $\Pi^{\overrightarrow{Q}^*}_{n}(\cdot|y^{n-1})$ are both Gaussian for fixed $(y^{n-1},x_n)$ and $y^{n-1}$, with conditional means which are linear in $(y^{n-1},x_n)$ and $y^{n-1}$, respectively, and conditional covariances which are independent of $(y^{n-1},x_n)$ and $y^{n-1}$, respectively.
\paragraph*{Stages $t \in \{n-1, n-2,\ldots,1, 0\}$} By (\ref{nonstationary:eq.31b}), evaluated at $t=n-1$, then
 $g_{n-1,n}(x_{n-1},y^{n-1})$ will include terms of quadratic form in $x_{n-1}$ and $y^{n-1}$. Repeating this argument recursively, it can be verified that at any time ${t\in\mathbb{N}_0^{n-1}}$, the optimal reproduction distribution $Q^*_{t}(\cdot|y^{t-1},x_t)$ is conditionally Gaussian with conditional means linear with respect to $(x_t,y^{t-1})$, and conditional covariances independent of $(x_t,y^{t-1})$,~${t\in\mathbb{N}_0^{n-1}}$. 
 
\par By induction, we then deduce that the optimal reproduction distributions are conditionally Gaussian, and they are realized using a general equation of the form 
\begin{align}
{Y}_t=\bar{A}_tX_t+\bar{B}_t{Y}^{t-1}+V^c_t,~{t\in\mathbb{N}_0^{n}}, \label{nonstationary:example:gaussian:eq.10}
\end{align}
where $\bar{A}_t\in\mathbb{R}^{p\times{p}}$, $\bar{B}_t\in\mathbb{R}^{p\times{t}p}$, and $\{V^c_t:~{t\in\mathbb{N}_0^{n}}\}$ is an independent sequence of Gaussian vectors $\{N(0;Q_t):~{t\in\mathbb{N}_0^{n}}\}$.

\par Next, we simplify the computation by introducing the following preprocessing at the encoder and decoder associated with channel (\ref{nonstationary:example:gaussian:eq.10}) (as shown in Fig.~\ref{nonstationary:communication_system}).\\
\noi{\it Preprocessing at Encoder.} Introduce (i) the estimation error $\{K_t:~{t\in\mathbb{N}_0^{n}}\}$ of $\{X_t:~{t\in\mathbb{N}_0^{n}}\}$ based on $\{Y_0,\ldots,Y_{t-1}\}$, and (ii) its covariance $\{\Pi_{t|t-1}:~{t\in\mathbb{N}_0^{n}}\}$, defined by
\begin{align}
K_t\triangleq{X}_t-\widehat{X}_{t|t-1},~\widehat{X}_{t|t-1}\triangleq\mathbb{E}\left\{X_t|\sigma\{{Y}^{t-1}\}\right\},~\Pi_{t|t-1}\triangleq\mathbb{E}\{K_tK_t^{\T}\},~{t\in\mathbb{N}_0^{n}},\label{nonstationary:example:gaussian:equation52}
\end{align}
where $\sigma\{{Y}^{t-1}\}$ is the $\sigma$-algebra (observable events) generated by the sequence $\{Y^{t-1}\}$. The covariance is diagonalized by introducing a unitary transformation $\{E_t: {t\in\mathbb{N}_0^{n}}\}$ such that 
\begin{align}
E_t\Pi_{t|t-1}{E}_t^{\T}=\Lambda_t,~\mbox{where}~\Lambda_t\triangleq\diag\{\lambda_{t,1},\ldots\lambda_{t,p}\},~{t\in\mathbb{N}_0^{n}}.\label{nonstationary:example:gaussian:equation53}
\end{align}
\noi To facilitate the computation, we introduce the scaling process $\{\Gamma_t:~{t\in\mathbb{N}_0^{n}}\}$, where $\Gamma_t\triangleq{E}_tK_t,~{t\in\mathbb{N}_0^{n}}$, has independent Gaussian components but  all of the components are correlated.\\
\noi{\it Preprocessing at Decoder.} Analogously, we introduce the error process $\{\tilde{K}_t:~{t\in\mathbb{N}_0^{n}}\}$ and the scaling process  $\{\widetilde{\Gamma}:~{t\in\mathbb{N}_0^{n}}\}$ defined by
\begin{align}
\tilde{K}_t\triangleq{Y}_t-\widehat{X}_{t|t-1},~\mbox{and}~\widetilde{\Gamma}_t{\triangleq}\Phi_tZ_t,~Z_t\triangleq\left(\Theta_t{E}_tK_t+V_t^c\right),~{t\in\mathbb{N}_0^{n}}.\label{nonstationary:example:gaussian:eq.12i}
\end{align}

\begin{figure}[!h]
\centering
\includegraphics[width=0.99\columnwidth]{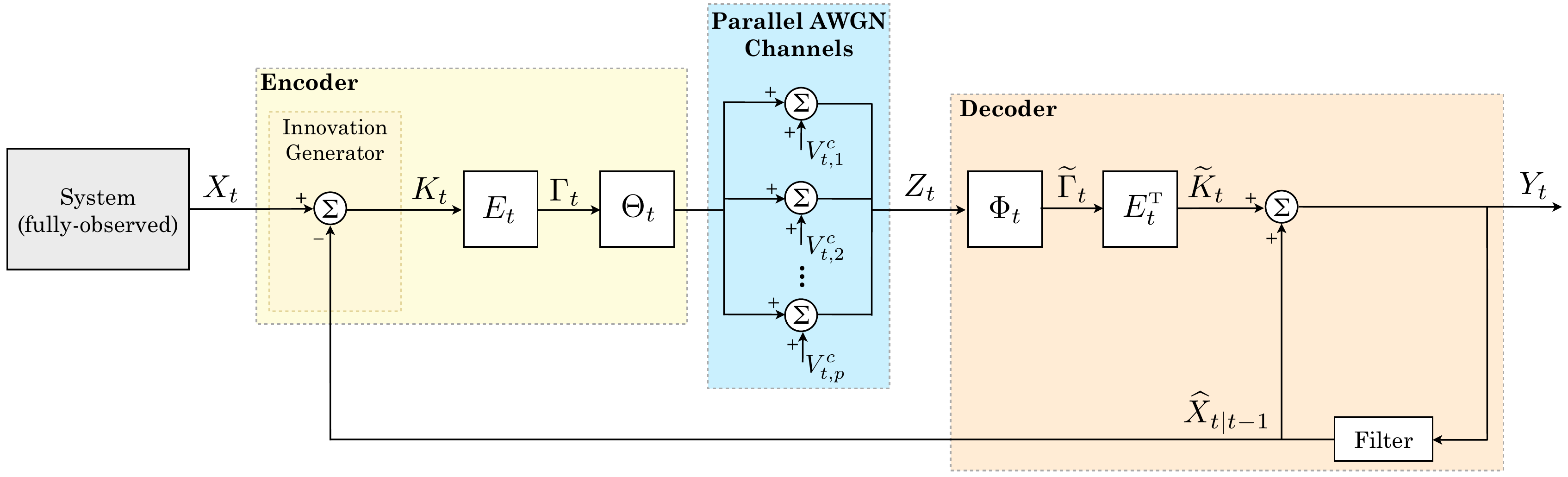}
\vspace{-0.1cm}
\caption{Realization of the optimal nonstationary reproduction distribution of multidimensional Gaussian process\vspace{-0.3cm}.}
\label{nonstationary:communication_system}
\end{figure}

\noindent The square error fidelity criterion $d_{0,n}(\cdot,\cdot)$ is not affected by the above processing of $\{(X_t,Y_t):~{t\in\mathbb{N}_0^{n}}\}$, since the preprocessing at both the encoder and decoder does not affect the form of the squared error distortion function, that is,
\begin{align}
\begin{split}
d_{0,n}(X^n,{Y}^n)&=d_{0,n}(K^n,\tilde{K}^n)=\frac{1}{n+1}\sum_{t=0}^n||\tilde{K}_t-K_t||_{2}^2 \\
&=\frac{1}{n+1}\sum_{t=0}^n||\widetilde{\Gamma}_t-\Gamma_t||_{2}^2=d_{0,n}(\Gamma^n,\widetilde{\Gamma}^n).\label{nonstationary:example:gaussian:equation.2abc}
\end{split}
\end{align}
\noi Using basic properties of conditional entropy, it can be shown that the following expressions are equivalent. 
\begin{align}
{R}_{0,n}^{na}(D)&=R_{0,n}^{na,K^n,\tilde{K}^n}(D)\triangleq\inf_{\{Q_{t}:~t=0,\ldots,n\}:~\mathbb{E}\left\{d_{0,n}(K^n,\tilde{K}^n)\right\}\leq{D}}\sum_{t=0}^n{I}(K_t;\tilde{K}_t|\tilde{K}^{t-1})\nonumber\\
&=R_{0,n}^{na,\Gamma^n,\widetilde{\Gamma}^n}(D)\triangleq\inf_{\{Q_{t}:~t=0,\ldots,n\}:~\mathbb{E}\left\{d_{0,n}(\Gamma^n,\widetilde{\Gamma}^n)\right\}\leq{D}}\sum_{t=0}^n{I}(\Gamma_t;\widetilde{\Gamma}_t|\widetilde{\Gamma}^{t-1}).\label{nonstationary:example:gaussian:equation.2}
\end{align}

\noi Next, we derive the main theorem which gives the closed form expression of the {NRDF} for multidimensional Gaussian process (\ref{nonstationary:example:gaussian:equation51}) by considering the feedback realization scheme shown in Fig.~\ref{nonstationary:communication_system}, where $\{V_c^t:~t\in\mathbb{N}_0^{n}\}$ is Gaussian $\{N(0;Q_t):~t\in\mathbb{N}_0^{n}\}$, and $\{{\Theta}_t,{\Phi}_t:~t\in\mathbb{N}_0^{n}\}$ are the matching matrices to be determined. 

\begin{theorem}($R_{0,n}^{na}(D)$ of time-varying multidimensional Gauss-Markov process)\label{nonstationary:example:gaussian:solution_gaussian} {\ \\}
{\bf (1)} The NRDF, $R_{0,n}^{na}(D)$, of the Gauss-Markov process (\ref{nonstationary:example:gaussian:equation51}), is given by
\begin{align}
{R}_{0,n}^{na}(D)&=\frac{1}{2}\frac{1}{n+1}\sum_{t=0}^n\sum_{i=1}^p\log\left(\frac{\lambda_{t,i}}{\delta_{t,i}}\right),~\delta_{t,i}\leq \lambda_{t,i},~t\in\mathbb{N}_0^{n},~i=1,\ldots,p,\label{nonstationary:example:gaussian:equa.13} \\
&\equiv\frac{1}{2}\frac{1}{n+1}\sum_{t=0}^n\sum_{i=1}^p\log\left\{\max\left(1,\frac{\lambda_{t,i}}{\delta_{t,i}}\right)\right\} ,\label{nonstationary:example:gaussian:determinant:expression} 
\end{align}
where  $\Lambda_t=E_t\Pi_{t|t-1}E^{\T}_{t}$,
\begin{align}
&\Pi_{t|t-1} \triangleq \mathbb{E}\left\{\left(X_t-\mathbb{E}\left\{X_t|\sigma\{Y^{t-1}\}\right\}\right)\left(X_t-\mathbb{E}\left\{X_t|\sigma\{Y^{t-1}\}\right\}\right)\T\right\}\label{nonstationary:example:gaussian:equa.133i}\\
&\delta_{t,i} \triangleq\left\{ \begin{array}{ll} \xi , & \mbox{if} \quad \xi\leq\lambda_{t,i} \\
\lambda_{t,i} , &  \mbox{if}\quad\xi>\lambda_{t,i} \end{array} \right.,~\forall{t,i}\label{nonstationary:example:gaussian:equa.13i}
\end{align}
and $\xi$ is chosen such that 
\begin{align}
\frac{1}{n+1}\sum_{t=0}^n\sum_{i=1}^p\delta_{t,i}=D.\label{average_distortion:eq.1}
\end{align}
{\bf (2)} The error $X_t-\mathbb{E}\{X_t|\sigma\{Y^{t-1}\}\}$ is Gaussian ${N}(0; \Pi_{t|t-1})$,~$\widehat{X}_{t|t-1}\triangleq\mathbb{E}\{X_t|\sigma\{Y^{t-1}\}\}$, and $\Pi_{t|t-1}$ are given by the Kalman filter equations
\begin{align}
\widehat{X}_{t+1|t}&=A_t\widehat{X}_{t|t-1}+A_t\Pi_{t|t-1}(E_t^{\T}H_{t}E_{t})^{\T}M_t^{-1}\left({Y}_t-\widehat{X}_{t|t-1}\right),\label{nonstationary:example:gaussian:10}\\
\Pi_{t+1|t}&=A_t\Pi_{t|t-1}{A}_t^{\T}-A_t\Pi_{t|t-1}(E_t^{\T}H_t{E}_{t})^{\T}M_{t}^{-1}(E_{t}^{\T}H_{t}E_{t})\Pi_{t|t-1}A_t^{\T}\nonumber\\
&+B_tB_{t}^{\T}=A_tE^{\T}_t\Delta_tE_tA_t^{\T}+B_tB_{t}^{\T},~\Pi_{0|-1}=\bar{\Pi}_{0|-1},~{t\in\mathbb{N}_0^{n}},\label{nonstationary:example:gaussian:11}\\
M_t&=E_t^{\T}H_t{E}_{t}\Pi_{t|t-1}(E_{t}^{\T}H_{t}E_{t})^{\T}+E_{t}^{\T}{\Phi}_{t}Q_t{\Phi}_{t}^{\T}E_t =E_t^{\T}H_t{\Lambda_t}{E}_t,\label{nonstationary:example:gaussian:11i}
\end{align}
where 
\begin{align}
\begin{split}
&\eta_{t,i}=1-\frac{\delta_{t,i}}{\lambda_{t,i}},~i=1,\ldots,p,~H_t\triangleq\diag\{\eta_{t,1},\ldots,\eta_{t,p}\},\\
&~\Delta_t=\diag\{\delta_{t,1},\ldots,\delta_{t,p}\},~{\Phi}_t\triangleq\sqrt{H_t\Delta_tQ_t^{-1}},~{t\in\mathbb{N}_0^{n}} .\label{nonstationary:example:gaussian:equation11mmm}
\end{split}
\end{align}
{\bf (3)} The realization of the optimal time-varying (nonstationary) reproduction distribution illustrated in Fig.~\ref{nonstationary:communication_system} is given by
\begin{align}
\begin{split}
Y_t&=E_t^{\T}{H}_t{E}_t\left(X_t-\widehat{X}_{t|t-1}\right)+E_t^{\T}{\Phi}_tV_t^c+\widehat{X}_{t|t-1}\\ 
&=\widehat{X}_{t|t-1}+E_t^{\T}{\Phi}_tZ_t,~~Z_t=\Theta_tE_t\left(X_t-\widehat{X}_{t|t-1}\right)+V_t^c,~~\Theta_t=\Phi^{-1}_tH_t.\label{nonstationary:10ab}
\end{split}
\end{align}
{\bf (4)} The filter estimate satisfies
\begin{align}
\widehat{X}_{t|t-1}&=A_{t-1}Y_{t-1},~\widehat{X}_{0|-1}=\mathbb{E}\{X_0|\sigma\{Y^{-1}\}\},~{t\in\mathbb{N}_0^{n}}\label{filter:outcome:eq:1}\\
\widehat{X}_{t|t}&=Y_t~\label{filter:outcome:eq:2}
\end{align}
and the optimal reproduction process is
\begin{align}
Y_t=A_{t-1}Y_{t-1}+E_t^{\T}\Phi_tZ_t,~ Z_t=\Theta_tE_t\left(X_t-A_{t-1}Y_{t-1}\right)+V^c_t. \label{nonstationary:10abbba}
\end{align}
{\bf (5)} The processes $\{Y_t:{t\in\mathbb{N}_0^{n}}\}$ and $\{\tilde{K}_t:{t\in\mathbb{N}_0^{n}}\}$ generate the same information, i.e., $\sigma\{Y^t\}=\sigma\{\tilde{K}^t\},~{t\in\mathbb{N}_0^{n}}$.
\end{theorem}
\begin{proof}
See Appendix~\ref{nonstationary:theorem:proof:solution_gaussian}.
\end{proof}

\par We make the following observations regarding Theorem~\ref{nonstationary:example:gaussian:solution_gaussian}.
\begin{remark}
\label{remark:comments}
\begin{itemize}
\item[{\bf (1)}] The main features of Theorem~\ref{nonstationary:example:gaussian:solution_gaussian} are {the following:} \\ First, by \eqref{nonstationary:10abbba} the information structure of the optimal reproduction for the specific Gaussian source with memory given by  \eqref{nonstationary:example:gaussian:equation51} is Markov, i.e.,
\begin{align}
{Q}^*_{t}(dy_t|y^{t-1},x^t)\equiv{Q}^*_{t}(dy_t|y_{t-1},x_t).\label{structure:reproduction}
\end{align}
Hence, the output process $\{Y_t:~t\in\mathbb{N}_0^n\}$ is first order Markov.\\
Second, the time-space reverse-waterfilling property \eqref{nonstationary:example:gaussian:equa.13}-\eqref{nonstationary:example:gaussian:equa.13i}, states that if the reproduction error $\delta_{t,i}$ is above the eigenvalue $\lambda_{t,i}$ of the error covariance $\Pi_{t|t-1}$, then the time-space component $X_{t,i}$\footnote{$X_{t,i}$ is the time-space component of the vector process $\{X_t:~t\in\mathbb{N}^n_{0}\}$.} is not reconstructed by $Y_{t,i}$\footnote{$Y_{t,i}$ is the time-space component of the vector process $\{Y_t:~t\in\mathbb{N}^n_{0}\}$.}, for $t\in\mathbb{N}_0^n,~i=1,\ldots,p$. The behavior of $\delta_{t,i}$ is described by the reverse-waterfilling expression \eqref{nonstationary:example:gaussian:equa.13i}, and the level $\xi$ depends on $D$, i.e., the overall fidelity of the error.
\item[{\bf (2)}] For each $t+1$, $\widehat{X}_{t+1|t}=A_tY_t$, given by \eqref{nonstationary:example:gaussian:10}, is the  estimator of $X_{t+1}$ based on $Y^t$. In addition, the time-space reverse-waterfilling is part of the estimation algorithm. This is a variant of the Kalman filter.
\end{itemize}
\end{remark}

\par The following remark, is a direct consequence of Theorem~\ref{nonstationary:example:gaussian:solution_gaussian}, and illustrates the connection between $R^{na}_{0,n}(D)$ and $D_{0,n}(R^{na})$ given by \eqref{section:ndrf:equation1}.

\begin{remark}
\label{corollary}{\ \\}
From Theorem~\ref{nonstationary:example:gaussian:solution_gaussian} the NRDF of the Gaussian process \eqref{nonstationary:example:gaussian:equation51} is given by 
\begin{align}
{R}_{0,n}^{na}(D)&=\frac{1}{2}\frac{1}{n+1}\sum_{t=0}^n\sum_{i=1}^p\log\left\{\max\left(1,\frac{\lambda_{t,i}}{\delta_{t,i}}\right)\right\}\stackrel{(a)}\equiv\frac{1}{n+1}\sum_{t=0}^n\sum_{i=1}^p{R}_{t,i}^{na}(\delta_{t,i})\label{proof:lower:bound:msed:equation1}
\end{align}
where $(a)$ follows if we let 
\begin{align}
{R}_{t,i}^{na}(\delta_{t,i})\triangleq\frac{1}{2}\log\left\{\max\left(1,\frac{\lambda_{t,i}}{\delta_{t,i}}\right)\right\},~t\in\mathbb{N}_0^n, ~i=1,\ldots,p.\label{proof:lower:bound:msed:equation2}
\end{align}
By \eqref{proof:lower:bound:msed:equation2} we obtain 
\begin{align}
\delta_{t,i}=\lambda_{t,i}e^{-2R^{na}_{t,i}},~t\in\mathbb{N}_0^n, ~i=1,\ldots,p.\label{corollary:equation1}
\end{align}
Utilizing \eqref{average_distortion:eq.1}, we obtain
\begin{align}
D=\frac{1}{n+1}\sum_{t=0}^n\delta_t=\frac{1}{n+1}\sum_{t=0}^n\sum_{i=1}^p\delta_{t,i},~\delta_t\triangleq\sum_{i=1}^p\delta_{t,i}.\label{corollary:equation2}
\end{align}
Substituting \eqref{corollary:equation1} into \eqref{corollary:equation2} we obtain
\begin{align}
D_{0,n}(R^{na})=\frac{1}{n+1}\sum_{t=0}^n\delta_t=\frac{1}{n+1}\sum_{t=0}^n\sum_{i=1}^p\lambda_{t,i}e^{-2R^{na}_{t,i}}.\label{corollary:equation3}
\end{align}
\end{remark}

\par Next, we utilize the closed form expressions of the NRDF and FTH NDRF evaluated for time-varying multidimensional Gauss-Markov process to derive a lower bound on the MSE given in terms of conditional mutual information $I(X^n;Y_0^n|Y^{-1})$.
\begin{theorem}(Universal lower bound on mean square error)\label{theorem:universal:bound}{\ \\}
Let $\{X_t:~t\in\mathbb{N}_0^n\}$ be the multidimensional Gauss-Markov process given by \eqref{nonstationary:example:gaussian:equation51} and let $\{\widetilde{Y}_t:~t\in\mathbb{N}_0^n\}$ be any estimator (not necessarily Gaussian) of $\{X_t:~t\in\mathbb{N}_0^n\}$. The mean square error is bounded below by
\begin{align}
\frac{1}{n+1}\sum_{t=0}^n\mathbb{E}\left\{||X_t-\widetilde{Y}_t||_2^2\right\}\geq\frac{1}{n+1}\sum_{t=0}^n\sum_{i=1}^p\lambda_{t,i}e^{-2I(X_{t,i};\widetilde{Y}_{t,i}|\widetilde{Y}_{t-1,i})}\label{equation:universal lower bound}
\end{align}
\end{theorem}
\begin{proof}
Let $D=\frac{1}{n+1}\sum_{t=0}^n\mathbb{E}\left\{||X_{t}-\widetilde{Y}_{t}||_2^2\right\}$ where 
$$\mathbb{E}\left\{||X_{t}-\widetilde{Y}_{t}||_2^2\right\}=\sum_{i=1}^p\delta_{t,i} \text{ with } D\in[0,\infty).
$$
Since, in general, $R_{t,i}^{na}\leq{I}(X_{t,i};\widetilde{Y}_{t,i}|\widetilde{Y}_{t-1,i}),~t\in\mathbb{N}_0^n,~i=1,\ldots,p$, then by \eqref{corollary:equation3}, we obtain
\begin{align}
\frac{1}{n+1}\sum_{t=0}^n\mathbb{E}\left\{||X_t-\widetilde{Y}_t||_2^2\right\}&={D}_{0,n}(R^{na})=\frac{1}{n+1}\sum_{t=0}^n\sum_{i=1}^p\lambda_{t,i}e^{-2R_{t,i}^{na}}\nonumber\\
&\geq\frac{1}{n+1}\sum_{t=0}^n\sum_{i=1}^p\lambda_{t,i}e^{-2I(X_{t,i};\widetilde{Y}_{t,i}|\widetilde{Y}_{t-1,i})}\label{proof:universal:lower:bound:msed:equation1} ,
\end{align}
which is the desired result. This completes the proof.
\end{proof}

\par Notice that from Remark~\ref{remark:comments}, ${\bf(2)}$, if we substitute $\widetilde{Y}_t=\widehat{X}_{t|t-1}=A_{t-1}Y_{t-1}$ in Theorem~\ref{theorem:universal:bound}, then we have the lower bound \eqref{equation:universal lower bound}.

\par In the next remark, we relate degenerated versions of the lower bound given by \eqref{equation:universal lower bound} to existing results in the literature.
\begin{remark}(Relations to existing results)\label{lower:bound:classical:rdf}
\begin{itemize}
\item[(a)] \cite[Theorem 5.8.1]{ihara1993},\cite{blahut1987} Let $X=(X_1,\ldots,X_p)$ be a $p-$dimensional Gaussian vector with distribution $X\sim N(0;\Gamma_X)$ and $Y=(Y_1,\ldots,Y_p)$ be its reproduction vector. Then, for any $D>0$,
\begin{align}
R(D)\triangleq\inf_{Q(dy|x):\mathbb{E}||X-Y||_2^2\leq{D}}I(X;Y)=\frac{1}{2}\sum_{i=1}^p\log\left\{\max\left(1,\frac{\lambda_i}{\xi}\right)\right\}\label{remark:equation1}
\end{align}
where $\{\lambda_i:~i=1,\ldots,p\}$ are the eigenvalues of $\Gamma_X$ and $\xi>0$ is a constant uniquely determined by $\sum_{i=1}^p\min\{\lambda_i,\xi\}=D$.
Note that the solution of classical RDF in \eqref{remark:equation1} is based on reverse-waterfilling method (see \cite[Lemma 5.8.2]{ihara1993}). The above results are also obtained from Theorem~\ref {nonstationary:example:gaussian:solution_gaussian}, if we assume model  \eqref{nonstationary:example:gaussian:equation51} generates an $\IID$ sequence $\{X_t:~t\in\mathbb{N}_0^n\}$ (by setting $A_t=0$, $B_t=I$). In such case, $\widehat{X}_{t|t-1}=\mathbb{E}X_t=0$ and $\Pi_{t|t-1}=\mathbb{E}X_tX^{\T}_t=\Gamma_X$.
\item[(b)]  Assume $X\sim{N}(0;\sigma^2_{X})$. By \cite[Theorem 1.8.7]{ihara1993} the following holds.
\begin{align*}
R(D)&=\min_{Q(dy|x):~\mathbb{E}||X-Y||_2^2\leq{D}}I(X;Y) = \frac{1}{2}\log\left\{\max\left(1,\frac{\sigma^2_X}{D}\right)\right\},~D\geq{0},\\
D(R)&=\min_{Q(dy|x):~I(X;Y)\leq{R}}\mathbb{E}\left\{||X-Y||_2^2\right\}=\sigma^2_Xe^{-2R}. 
\end{align*} 
The realization scheme to achieve the classical RDF or the DRF is the following.
\begin{align}
Y=\Big(1-\frac{D}{\sigma^2_X}\Big)X+V^c,~V^c\sim{N}\left(0;D(1-\frac{D}{\sigma^2_X})\right).\label{remark:equation4}
\end{align}
Note that \eqref{remark:equation4} is a degenerated version of \eqref{nonstationary:10ab} assuming the model of \eqref{nonstationary:example:gaussian:equation51} generates $\IID$ sequence $\{X_t:~t\in\mathbb{N}_0^n\}$ as in {(a)}, and the connection to Theorem \ref {nonstationary:example:gaussian:solution_gaussian} is established by setting $E_t=1$, ${H}_t=1-\frac{D}{\sigma^2_X}$, $\widehat{X}_{t|t-1}=0$, ${\Phi}_t=H_tD$ and $V_t^c\sim{N}(0;1)$. 
\item[(c)](Lower bound on MSE \cite[1.8.8]{ihara1993},\cite{liptser-shiryaev2001}) Given a Gaussian RV $X\sim{N}(0;\sigma^2_X)$, then for any real valued RV $\widetilde{Y}$(not necessarily Gaussian) the MSE is bounded below by
\begin{align}
\mathbb{E}||X-\widetilde{Y}||_2^2\geq\sigma^2_Xe^{-2I(X;\bar{Y})}.\label{remark:equation2}
\end{align}
\end{itemize}
\end{remark}
The RDF of the Gaussian RV $X\sim{N}(0;\sigma^2_X)$ and the lower bound in \eqref{remark:equation2}, are utilized in \cite{ihara1993,liptser-shiryaev2001} to derive optimal coding and decoding schemes for transmitting a Gaussian message $\theta\sim{N}(0;\sigma^2_\theta)$ over an AWGN channel with feedback, $Y_t=X_t(\theta,Y^{t-1})+V_t^c,~t\in\mathbb{N}_0^n$, where $\{V_t^c:~t\in\mathbb{N}_0^n\}$ is $\IID$ Gaussian process. Although we do not pursue such problems in this paper, we note that Theorems \ref{nonstationary:example:gaussian:solution_gaussian} and \ref{theorem:universal:bound} are necessary in order to derive optimal coding schemes for additive Gaussian channels with memory (including additive Gaussian memoryless channels). 

\subsection{Examples}\label{examples}

\par In this section, we numerically compute the {NRDF} of {\it time-varying} Gauss-Markov process, using Theorem~\ref{nonstationary:example:gaussian:solution_gaussian}. For these examples, the utility of the reverse-waterfilling algorithm is necessary even when the process elements are scalar (i.e., $p=1$). For process elements in higher dimensions (i.e., $p\geq{2}$), the complexity of the problem increases, since the reverse-waterfilling algorithm must be solved both in time and space units. We overcome this obstacle by proposing an iterative algorithmic technique that allocates information of the Gaussian process and distortion levels optimally.
 
\begin{remark}(Relations to existing results){\ \\}
The examples presented here deal with the time-space aspects of the reverse-waterfilling algorithm. This is fundamentally different from \cite[Section IV.C]{stavrou-kourtellaris-charalambous2016ieeetit} where it is assumed that the optimal reproduction distributions $\{Q^*_t(dy_t|y^{t-1},x_t)=Q^*(dy_t|y^{t-1},x_t): t\in\mathbb{N}_0\}$ are time-invariant (identical).  
\end{remark}

\begin{varalgorithm}{1}
\caption{Rate distortion allocation algorithm: The vector case}
\begin{algorithmic}
\STATE {\textbf{Initialize:} }
\STATE {The number of time-steps $n$; the number of channels $p$ the distortion level $D$; the error tolerance $\epsilon$; the initial covariance matrix $\bar{\Pi}_{0|-1}$ of the error process $K_0$, the state-space matrices $A_t$ and $B_t$ of the time-varying multidimensional Gauss-Markov process $X_t$ given by \eqref{nonstationary:example:gaussian:equation51}.}
\STATE {}
\STATE {Set $\xi=D$; $\text{flag}=0$.}
\STATE {}
\WHILE{$\text{flag}=0$}
\STATE {Compute $\delta_{t,i}~\forall~t,i$ as follows:}
\FOR {$t=0:n$}
\STATE {Perform Singular Value Decomposition: $[E_t,\Lambda_t]=\texttt{SVD}(\Pi_{t|t-1})$}
\STATE {$\Delta_t$ is computed according to \eqref{nonstationary:example:gaussian:equa.13i}.}
\STATE {Use $A_t$ $B_t$ and $\Delta_t$ to compute $\Pi_{t+1|t}$ according to \eqref{nonstationary:example:gaussian:11}.}
\ENDFOR
\IF {$|\frac{1}{n+1}\sum_{t=0}^n\sum_{i=1}^p\delta_{t,i}-D| \leq \epsilon$}
\STATE {$\text{flag}\leftarrow 1$}
\ELSE
\STATE {Re-adjust $\xi$ as follows:}
\STATE {$\xi \leftarrow \xi +\beta(D-\frac{1}{n+1}\sum_{t=0}^n\sum_{i=1}^p\delta_{t,i})$, where $\beta\in(0,1]$ is a proportionality gain and affects the rate of convergence.}
\ENDIF
\ENDWHILE
\end{algorithmic}
\label{algo2}
\end{varalgorithm}

\begin{example}
Consider the following two-dimensional Gauss-Markov process
\begin{align}
\begin{bmatrix}
    X_{t+1,1} \\
    X_{t+1,2}
 \end{bmatrix}
=\underbrace{\begin{bmatrix}
    -\alpha_t & 1 \\
    -\beta_t & 0 
 \end{bmatrix}}_{A_t}
\begin{bmatrix}
    X_{t,1} \\
    X_{t,2}
 \end{bmatrix}
+\underbrace{\begin{bmatrix}
    \sigma_{W_{t,1}} & 0 \\
    0 & \sigma_{W_{t,2}} 
 \end{bmatrix}}_{B_t} 
\begin{bmatrix}
    W_{t,1} \\
    W_{t,2}
 \end{bmatrix}  
~t=0,1,2,~i=1,2,  \label{nonstationary:two:dimension:equation1}
\end{align}
where $W_{t,i}\sim{N}(0;1)$, $\sigma_{W_{t,i}}W_{t,i}\sim{N}(0;\sigma^2_{W_{t,i}})$ and $\{A_t,B_t\}$ are time-varying matrices. This example corresponds to \eqref{nonstationary:example:gaussian:equation51} for $p=k=n=2$.
For this example, we choose the distortion level $D=3$ and consider the following matrices $\{A_t,B_t\}$:
\begin{align*}
&A_0=\begin{bmatrix}
    -0.5 & 1 \\
    -0.4 & 0 
 \end{bmatrix}
,~B_0=\begin{bmatrix}
     1 & 0 \\ 
     0 & 1
\end{bmatrix} \\
&A_1=\begin{bmatrix}
    -0.4 & 1 \\
    -0.5 & 0 
 \end{bmatrix}
,~B_1=\begin{bmatrix}
     0.9 & 0 \\ 
     0 & 1.4
\end{bmatrix} \\
&A_2=\begin{bmatrix}
    -0.9 & 1 \\
    -0.5 & 0 
 \end{bmatrix}
,~B_2=\begin{bmatrix}
     1.2 & 0 \\ 
     0 & 1.3
\end{bmatrix}.
\end{align*}
The initial covariance matrix of the error process $K_t$ is
\begin{align*}
\bar{\Pi}_{0|-1}=\begin{bmatrix}
    0.6 & 0.2 \\ 
    0.2 & 0.4
\end{bmatrix}.
\end{align*}
Recall that the covariance matrix of the error process $K_t$ given by \eqref{nonstationary:example:gaussian:11} is simplified to 
\begin{align}
\Pi_{t+1|t}=A_tE^{\T}_t\big\{\diag\{\delta_{t,1},\delta_{t,2}\}\big\}E_tA^{\T}_t+B_tB^{\T}_t,~t=0,1,2,~\Pi_{0|-1}=\bar{\Pi}_{0|-1}\label{nonstationary:twodimension:example:equation1a}
\end{align}
and $\delta_{t,i}$ given by \eqref{nonstationary:example:gaussian:equa.13i} becomes
\begin{align} 
\delta_{t,i}=\min\{\lambda_{t,1},\xi\},~t=0,1,2,~i=1,2.\label{nonstationary:twodimension:example:equation1b}
\end{align}
Now let us implement Algorithm~\ref{algo2} for error tolerance $\epsilon=10^{-3}$. We choose an initial $\xi=\xi_0$ to start our iterations. A good starting point is $\xi_0=D$. 
For $\bar{\Pi}_{0|-1}$ we perform Singular Value Decomposition (SVD) and we obtain the unitary matrix 
\begin{align*}
E_0=
\begin{bmatrix}
   -0.8507   & -0.5257\\
  -0.5257   &  ~ 0.8507
\end{bmatrix}
\end{align*}
and the eigenvalues in a diagonal matrix that correspond to the levels of the noise $\lambda_{0,1}$ and $\lambda_{0,2}$, i.e., 
\begin{align*}
\Lambda_0=
\begin{bmatrix}
   0.7236   &      0\\
         0  &  0.2764
\end{bmatrix}.
\end{align*}
For $\xi=\xi_0= D=3$ and $(\lambda_{0,1},\lambda_{0,2})=(0.72, 0.28)$ we compute $\Delta_0$ using \eqref{nonstationary:twodimension:example:equation1b}. Hence,
\begin{align*}
\Delta_0=\Lambda_0=
\begin{bmatrix}
    0.7236   &      0\\
         0  &  0.2764
\end{bmatrix}.
\end{align*}

Using $A_0$, $B_0$, $\Delta_0$ and $E_0$ we compute $\Pi_{1|0}$ using \eqref{nonstationary:twodimension:example:equation1a} 
\begin{align*}
{\Pi}_{1|0}=\begin{bmatrix}
1.3500     & 0.0400\\
    0.0400 & 1.0960
\end{bmatrix},
\end{align*}
and the procedure of (a) computing the SVD of $\Pi_{1|0}$, (b) computing $\Delta_1$ is repeated as it is done for ${\Pi}_{1|0}$. Similarly, the procedure is repeated for all $t=0,1,\ldots, n$. At the end, for the given $\xi$ we check if $|\frac{1}{n+1}\sum_{t=0}^n\sum_{i=1}^p\delta_{t,i}-D| \leq \epsilon$. If it does, we stop the iterations and the last $\xi$ is the level we want. If not, we update $\xi$ as $\xi \leftarrow \xi +\beta(D-\frac{1}{n+1}\sum_{t=0}^n\sum_{i=1}^p\delta_{t,i})$ and we repeat the procedure for all $t$ again. For this example, the final reverse-waterfilling is found in 9 iterations and it is shown in Figure~\ref{fig:example1_reverseWaterfilling}.
\begin{figure}[h]
\centering
\includegraphics[width=0.6\columnwidth]{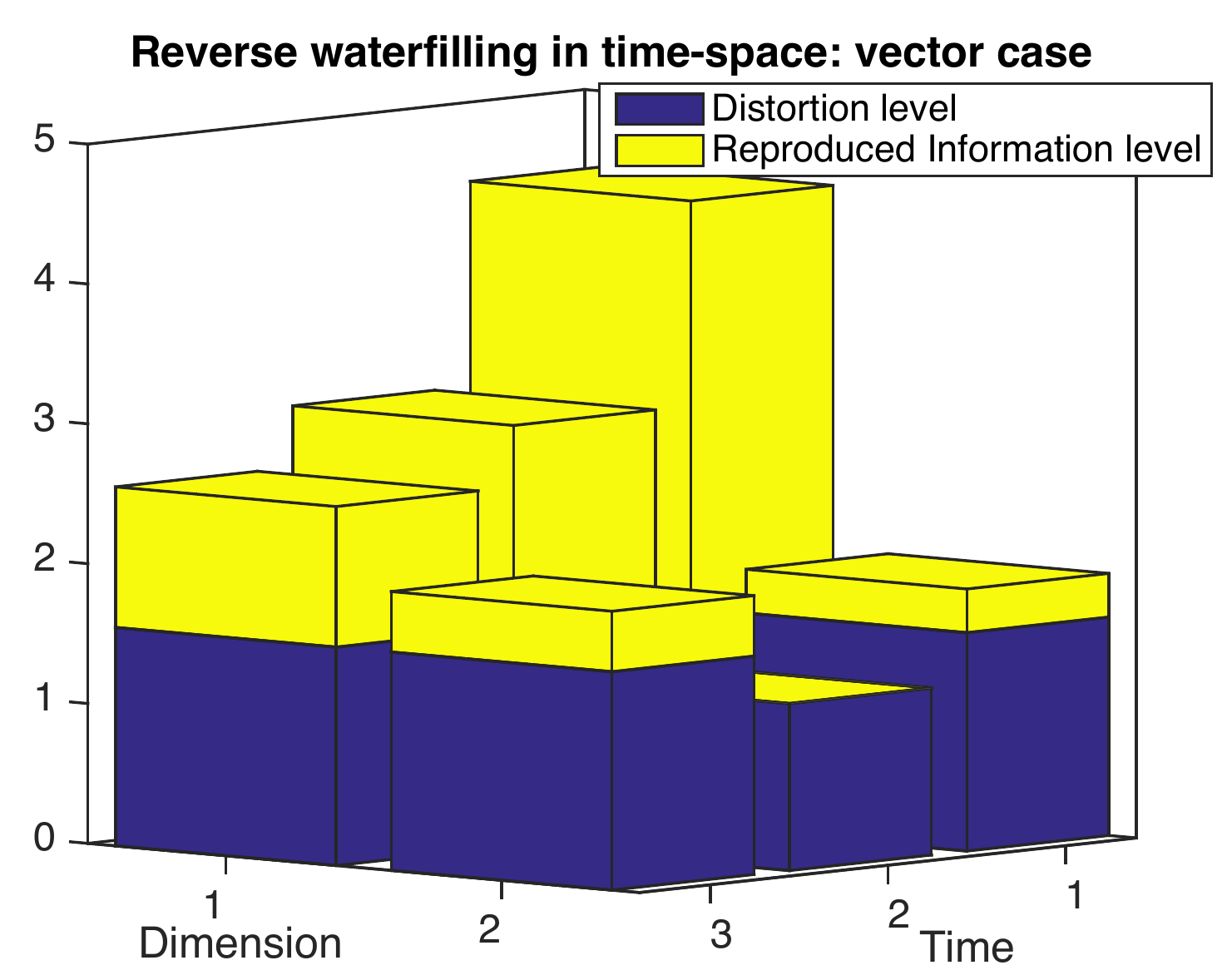}
\caption{Reverse-waterfilling in time-space for n=2 time units and p=2 space units.}
\label{fig:example1_reverseWaterfilling}
\end{figure}

By (\ref{nonstationary:example:gaussian:equa.13}) we compute the {NRDF}: 
\begin{align*}
{R}_{0,2}^{na}(D)=\frac{1}{2}\frac{1}{2+1}\sum_{t=0}^2\sum_{i=1}^2\log\left(\frac{\lambda_{t,i}}{\delta_{t,i}}\right) 
=0.6330~\mbox{bits/source symbol}
\end{align*}

\end{example}

In the next corollary, we degrade the results derived in Theorem~\ref{nonstationary:example:gaussian:solution_gaussian} to the case of {\it time-varying scalar Gauss-Markov process}. This corollary emphasizes on the fact that even in its simplest form, i.e., when $p=1$, the computation of FTH NRDF for time-varying Gauss-Markov processes can only be evaluated numerically by utilizing algorithmic methods. Note that in the sequel, when we refer to the scalar Gaussian process, for simplicity we will not make use of the dimension subscript, that is, $\lambda_{t,1}\equiv\lambda_t, \delta_{t,1}\equiv\delta_t, \eta_{t,1}\equiv\eta_t, q_{t,1}\equiv{q}_t$ etc.

\begin{corollary}($R^{na}_{0,n}(D)$ of time-varying scalar Gauss-Markov process)\label{nonstationary:example:gaussian:scalar:solution_gaussian}{\ \\}
This corresponds to (\ref{nonstationary:example:gaussian:equation51}) by setting $p=k=1$, $A_t=\alpha_t$, $B_t=\sigma_{W_{t}}$, i.e., $\sigma_{W_{t}}W_t\sim{N}(0;\sigma^2_{W_{t}})$ giving
\begin{align}
X_{t+1}=\alpha_t{X}_t+\sigma_{W_{t}}{W}_t,~W_t\sim{N}(0;1),~X_0\sim{N}(0;\sigma^2_{X_{0}}), ~t=0,1,\ldots,n\label{nonstationary:scalar:equation111}
\end{align}
where $\{\alpha_t,\sigma_{W_t}:~t=0,1,\ldots,n\}$ are time varying. Then $\sigma^2_{X_{t}}\triangleq\var(X_t)$, satisfies~$\sigma^2_{X_{t+1}}=\alpha^2_t\sigma^2_{X_t}+\sigma^2_{W_t},~\sigma^2_{X_0}=\sigma^2_0,~t\in\mathbb{N}_0^n$. \\
In this case, by Theorem~\ref{nonstationary:example:gaussian:solution_gaussian}, and (\ref{nonstationary:example:gaussian:equa.13}) we obtain 
\begin{align}
{R}_{0,n}^{na}(D)=\frac{1}{2}\frac{1}{n+1}\sum_{t=0}^n\log\left(\frac{\lambda_{t}}{\delta_{t}}\right)\label{nonstationary:example:gaussian:scalar_markov1}
\end{align}
where
\begin{align}
\delta_{t} \triangleq\left\{ \begin{array}{ll} \xi & \mbox{if} \quad \xi\leq\lambda_{t} \\
\lambda_{t} &  \mbox{if}\quad\xi>\lambda_{t} \end{array} \right.,~\forall{t}\label{nonstationary:watefilling:scalar}
\end{align}
with $\xi$ fixed such that $\frac{1}{n+1}\sum_{t=0}^n\delta_{t}=D, \delta_{t}=\min_{t}\{\lambda_{t},\xi\}$ and $\Pi_{t|t-1}=\Lambda_t=\lambda_{t}$, (i.e., $E_t=1$), $H_t=\eta_{t}=1-\frac{\delta_{t}}{\lambda_{t}},~t=0,\ldots,n$.\\
By (\ref{nonstationary:example:gaussian:11i}), we obtain
\begin{align}
&M_{t}=\lambda_{t}H^2_t+H_t{\delta_{t}}=H_t\left(\lambda_{t}{H}_t+\delta_{t}\right)=H_t\left(\lambda_{t}\left(1-\frac{\delta_{t}}{\lambda_{t}}\right)+\delta_{t}\right)=\lambda_{t}{H}_t.\label{nonstationary:example:gaussian:11j}
\end{align}
Also, by (\ref{nonstationary:example:gaussian:11}), we obtain
\begin{align}
\lambda_{t+1}&=\alpha_t^2\lambda_{t}-\alpha_t^2\lambda^2_{t}{H}^2_t{M}^{-1}+\sigma_{W_t}^2\stackrel{(a)}
=\alpha_t^2\lambda_{t}-\alpha_t^2\lambda^2_{t}{H}^2_t{H}^{-1}_t\lambda^{-1}_{t}+\sigma_{W_t}^2\nonumber\\
&=\alpha_t^2\lambda_{t}-\alpha_t^2\lambda_{t}{H}_t+\sigma_{W_t}^2=\alpha_t^2\lambda_{t}-\alpha_t^2\lambda_{t}\left(1-\frac{\delta_{t}}{\lambda_{t}}\right)+\sigma_{W_t}^2=\alpha_t^2{\delta_{t}}+\sigma_{W_t}^2,~\bar{\lambda}_{0}=\sigma^2_{X_0}\label{nonstationary:example:gaussian:11k}
\end{align}
where $(a)$ follows from (\ref{nonstationary:example:gaussian:11j}).
\end{corollary}

Similarly to Algorithm~\ref{algo2}, we structure Algorithm~\ref{algo1} for rate distortion allocation.

\begin{varalgorithm}{2}
\caption{Rate distortion allocation algorithm: The scalar case}
\begin{algorithmic}
\STATE {\textbf{Initialize:} }
\STATE{The number of time-steps $n$; the distortion level $D$; the error tolerance $\epsilon$; the initial variance $\bar{\lambda}_0=\sigma^2_{X_{0}}$ of the initial state $X_0$, the values $a_t$ and $\sigma^2_{W_t}$ of the time-varying scalar Gauss-Markov process $X_t$ given by \eqref{nonstationary:scalar:equation111}}. 
\STATE{}
\STATE {Set $\xi=D$; $\text{flag}=0$.}
\STATE {}
\WHILE{$\text{flag}=0$}
\STATE {Compute $\delta_{t}~\forall~t$ as follows:}
\FOR {$t=0:n$}
\STATE {$\delta_t$ is computed according to \eqref{nonstationary:watefilling:scalar}.}
\STATE {Use $a_t$ and $\sigma^2_{W_t}$ to compute $\lambda_{t+1}$ according to \eqref{nonstationary:example:gaussian:11k}.}
\ENDFOR
\IF {$|\frac{1}{n+1}\sum_{t=0}^n\delta_{t}-D| \leq \epsilon$}
\STATE {$\text{flag}\leftarrow 1$}
\ELSE
\STATE {Re-adjust $\xi$ as follows:}
\STATE {$\xi \leftarrow \xi +\beta(D-\frac{1}{n+1}\sum_{t=0}^n\delta_{t})$, where $\beta\in(0,1]$ is a proportionality gain and affects the rate of convergence.}
\ENDIF
\ENDWHILE
\end{algorithmic}
\label{algo1}
\end{varalgorithm}

\begin{example}\label{nonstationary:nrdf:example:scalar:waterfilling}
For this example, we choose the distortion level $D=2$ and use the following $\{a^2_t,\sigma^2_{W_t}\}$:
\begin{align*}
&(a^2_0, \sigma^2_{W_0})=(1,1) , \quad (a^2_1, \sigma^2_{W_1})=(0.2,1.3) , \quad (a^2_2, \sigma^2_{W_2})=(1.8,0.7) .
\end{align*}
The initial variance is $\sigma_{X_0}=1$. Hence, $\bar{\lambda}_0=\sigma_{X_0}=1$.

Now let us implement Algorithm~\ref{algo1} for error tolerance $\epsilon=10^{-3}$. We choose an initial $\xi=\xi_0$ to start our iterations. A good starting point is $\xi_0=D$. Using \eqref{nonstationary:watefilling:scalar}, $\delta_0=\min\{1,2\}=1$. Then, using \eqref{nonstationary:example:gaussian:11k}, $\lambda_1=\alpha_0^2{\delta_{0}}+\sigma_{W_0}^2$ and thus $\delta_1$ is computed. Similarly, the procedure is repeated for all $t=0,1,\ldots, n$. At the end, for the given $\xi$ we check if $|\frac{1}{n+1}\sum_{t=0}^n\delta_{t}-D| \leq \epsilon$. If it does, we stop the iterations and the last $\xi$ is the level we want. If not, we update $\xi$ as $\xi \leftarrow \xi +\beta(D-\frac{1}{n+1}\sum_{t=0}^n\delta_{t})$ and we repeat the procedure for all $t$ again.

For this example, the final reverse-waterfilling is found after $15$ iterations and it is shown in Figure~\ref{fig:example1_reverseWaterfilling}.
\begin{figure}[h]
\centering
\includegraphics[width=0.6\columnwidth]{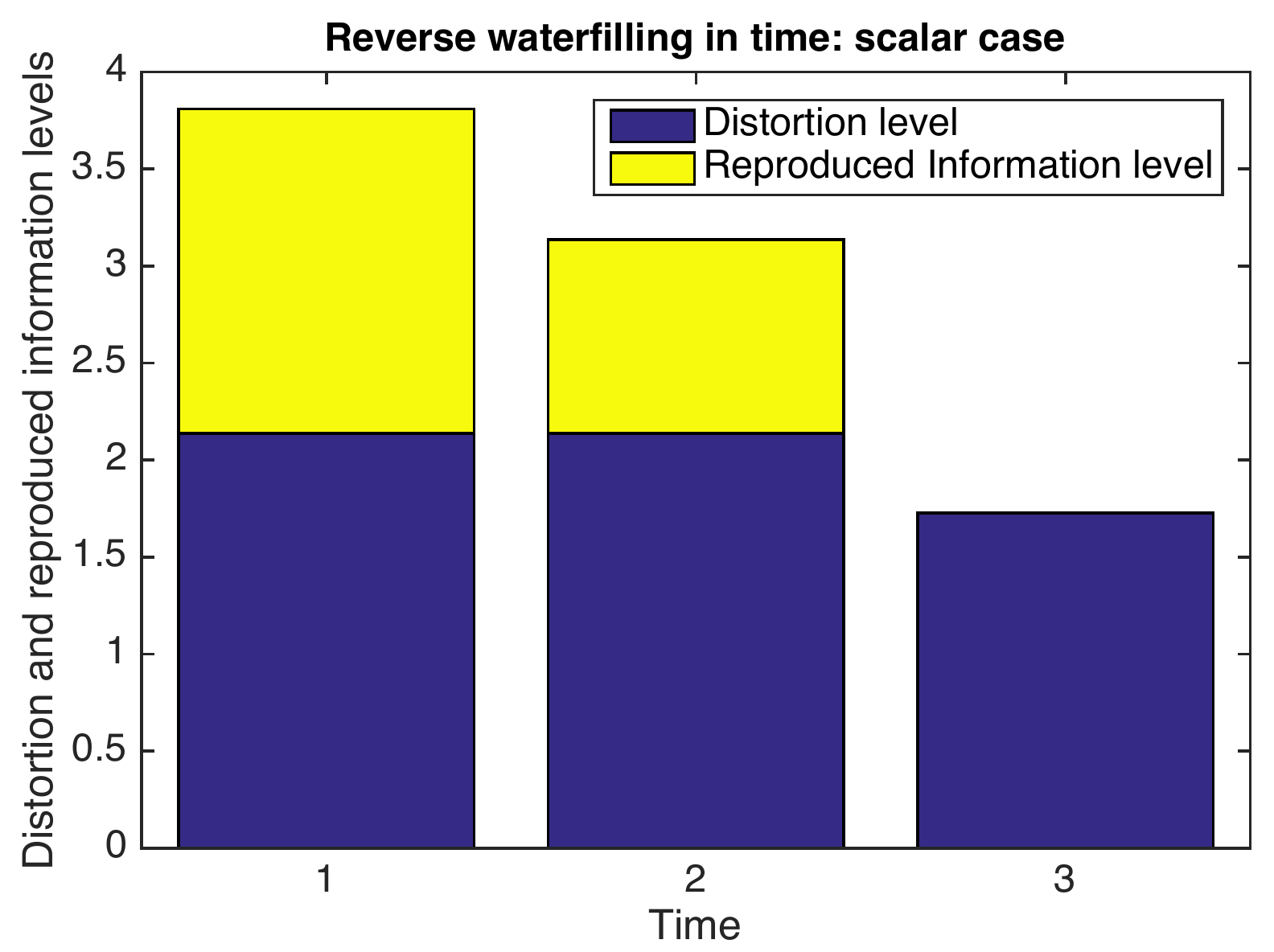}
\caption{Reverse-waterfilling in time for n=2 time units.}
\label{fig:example2_scalar}
\end{figure}

By \eqref{nonstationary:example:gaussian:scalar_markov1} we compute the NRDF: 
\begin{align*}
{R}_{0,2}^{na}(D)&=\frac{1}{2}\frac{1}{2+1}\sum_{t=0}^2\log\left(\frac{\lambda_{t}}{\delta_{t}}\right) 
=0.2314~\mbox{bits/source symbol}
\end{align*}

\end{example}

\subsection{Realization of \eqref{nonstationary:10abbba}}

In this section, we exemplify the relation between information-based estimation via NRDF and the fact that the latter can also be seen as a realization of an \texttt{\{encoder, channel, decoder\}} processing information optimally with zero-delay. For simplicity we consider scalar process ($p=1$). Note that this concept is precisely the one described in Fig.~\ref{introduction:fig:realization:nrdf:zero-delay}.

\begin{example}(Realization of \eqref{nonstationary:10abbba} for scalar processes)\label{example:gaussian:scalar:feeedback}{\ \\}
Let $X_t$ be the scalar time-varying Gauss-Markov process defined by (\ref{nonstationary:scalar:equation111}), and recall that  $\lambda_{t}=\alpha_{t-1}^2{\delta_{t-1}}+\sigma^2_{W_{t-1}}$ and $R^{na}(D)=\frac{1}{2}\sum_{t=0}^n\log\left(\frac{\delta_{t,i}}{\lambda_{t,i}}\right)$ (see (\ref{nonstationary:example:gaussian:11k}) and \eqref{nonstationary:example:gaussian:scalar_markov1}, respectively). 
\par Using \eqref{nonstationary:10abbba} for $p=1$, we obtain the following expression:
\begin{align}
Y_t=\Phi_tZ_t+\alpha_{t-1}Y_{t-1},~ Z_t=\Theta_t\left(X_t-\alpha_{t-1}Y_{t-1}\right)+V^c_t\label{example:realization:scalar},
\end{align}
where 
\begin{align}
{\Phi}_t\triangleq\sqrt{\frac{H_t\delta_t}{q_t}}~~~\mbox{and}~~~\Theta_t\triangleq\sqrt{\frac{H_tq_t}{\delta_t}}.\label{example:scalings}
\end{align}
Note that $H_t, \delta_t$ are defined in Corollary~\ref{nonstationary:example:gaussian:scalar:solution_gaussian} and $q_t$ is the variance of the scalar noise process $V^c_t~\sim{N}(0;q_t)$.
\par Next,  we consider the FTH information capacity of a memoryless AWGN channel with or without feedback with Gaussian noise process given as follows
\begin{align}
C_{0,n}(P)=\frac{1}{2}\frac{1}{n+1}\sum_{t=0}^n\log\left(1+\frac{P_t}{q_t}\right)\label{gaussian_capacity_scalar} ,
\end{align}
where $P_t$ is the power level allocated at each time.\\
Suppose that this channel is used once per source symbol, that is, the coding rate between the source symbols and the channel symbols is $1$ \cite[Definition 2.1]{gastpar2002}. For the realization in \eqref{example:realization:scalar}, the smallest achievable distortion is obtained by setting \eqref{nonstationary:example:gaussian:scalar_markov1}=\eqref{gaussian_capacity_scalar} that yields
\begin{align}
{\delta^{\min}_{t}}=\frac{\lambda_{t}q_t}{q_t+P_{t}}=\frac{\left(\alpha_{t-1}^2{\delta^{\min}_{t-1}}+\sigma^2_{W_{t-1}}\right)q_t}{q_t+P_{t}},~{t\in\mathbb{N}_1^{n}} ,\label{nonstationary:scalar:minimum_distortion_2}
\end{align}
where
\begin{align}
D_{\min}=\frac{1}{n+1}\sum_{t=0}^n\delta^{\min}_{t}. \label{example:jscc:scalar:eq.3}
\end{align}
Evaluating  $\lambda_{t}=\alpha_{t-1}^2{\delta_{t-1}}+\sigma^2_{W_{t-1}}$ at $\delta_t^{\min}$ the following feedback encoder operates at FTH information capacity.
\begin{align}
\begin{split}
Z_t&= \sqrt{\frac{P_t}{\lambda_t}}K_t+q_t=\sqrt{\frac{P_t}{\alpha_{t-1}^2{\delta^{\min}_{t-1}}+\sigma^2_{W_{t-1}}}}K_t+q_t, \\
K_t&=X_t-\mathbb{E}\{X_t|\sigma\{Y^{t-1}\}\} = X_t-\alpha_{t-1}{Y}_{t-1}\label{nonstationary:additive_channel_feedback1b}  ,
\end{split}
\end{align}
where $Z_t$ is the observation process containing the data.\\
In addition, the decoder (or the filter) is given by the realization in \eqref{example:realization:scalar}.
\par By \eqref{example:scalings}, the scaling factor ${\Phi}_t$ which guarantees the minimum end-to-end error is 
\begin{align}
{\Phi}_t=\sqrt{\frac{\alpha^2_{t-1}\delta^{\min}_{t-1}+\sigma^2_{W_{t-1}}}{P_{t}}}\frac{P_{t}}{q_t+P_{t}}.\label{scaling_factor_feedback_2}
\end{align}
Substituting (\ref{scaling_factor_feedback_2}) into (\ref{example:realization:scalar}) we obtain 
\begin{align}
Y_t = \alpha_{t-1}Y_{t-1} +\sqrt{\frac{\alpha^2_{t-1}\delta^{\min}_{t-1}+\sigma^2_{W_{t-1}}}{P_{t}}}\frac{P_{t}}{q_t+P_{t}}{Z}_t.
\label{decoder2b}
\end{align}
Finally, the average end-to-end distortion at each time instant is computed by evaluating the expectation
\begin{align*}
D_{t}=\mathbb{E}\{|X_t-Y_t|^2\}&=\frac{(\alpha^2_{t-1}\delta^{\min}_{t-1}+\sigma^2_{W_{t-1}})q^2_t+(\alpha^2_{t-1}\delta^{\min}_{t-1}+\sigma^2_{W_{t-1}})q_tP_{t}}{(P_{t}+q_t)^2}\nonumber\\
&=\frac{(\alpha^2_{t-1}\delta^{\min}_{t-1}+\sigma^2_{W_{t-1}})q_t}{q_t+P_{t}}=\delta^{\min}_{t}.
\end{align*} 
\noi The realization of \eqref{example:realization:scalar} with an \texttt{\{encoder, channel, decoder\}} operating with zero-delay is illustrated in Fig.~\ref{fig:nonstationary:realization_of_filter}.
\begin{figure}[ht]
\centering
\includegraphics[width=\columnwidth]{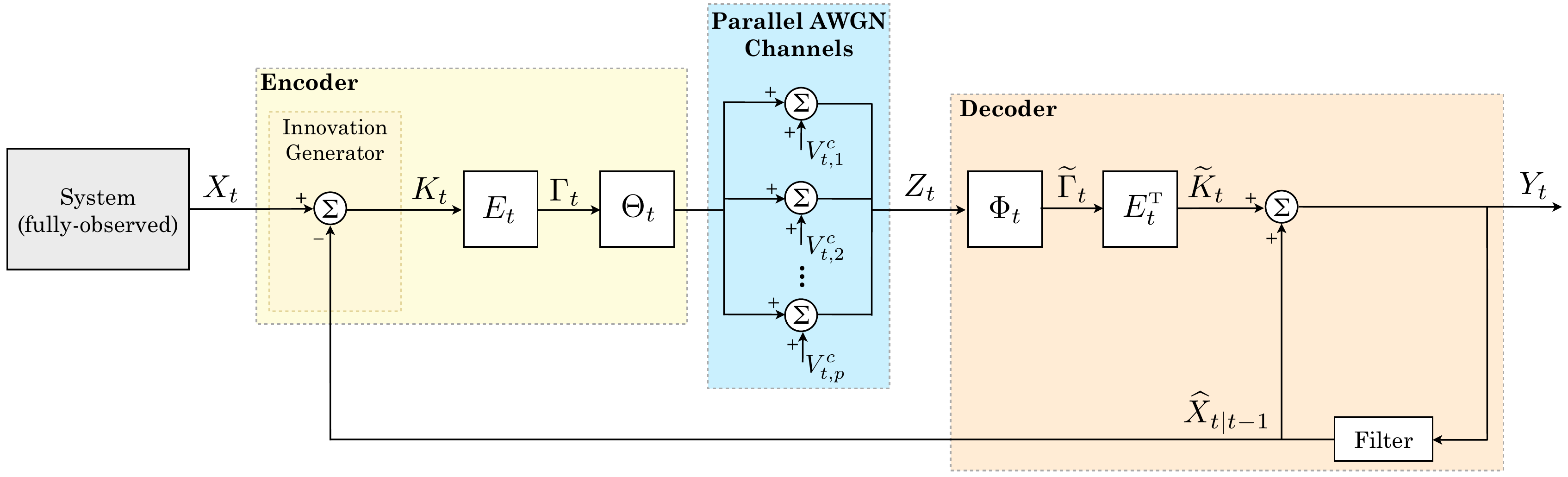}
\caption{Realization of the optimal reproduction of $R_{0,n}^{na}(D)$ given by \eqref{example:realization:scalar}. The scalings $\Theta_t$ and $\Phi_t$ are given by \eqref{nonstationary:additive_channel_feedback1b} and \eqref{scaling_factor_feedback_2} respectively\vspace{-0.3cm}.}\label{fig:nonstationary:realization_of_filter}
\end{figure}

\end{example}

%
%
%
%
\section{Conclusions and Future Directions}\label{sec:conclusions}

\par In this paper, we derived information-based causal filters via nonanticipative rate distortion theory in finite-time horizon. We exemplified our theoretical framework to time-varying multidimensional Gauss-Markov process subject to a MSE fidelity, and we demonstrated that obtaining such filters is equivalent to the design of an optimal \texttt{\{encoder, channel, decoder\}}, which ensures that the error fidelity is met. Unlike classical Kalman filters, the new information-based causal filter is characterized by a reverse-waterfilling algorithm. Moreover, we established a universal lower bound on the MSE of any estimator of a Gaussian random process. 

\par The results derived in this paper makes pave the way to generalizing the proposed framework to Gaussian sources governed by partially observed Gauss-Markov processes. Part of ongoing research focuses on how filtering with fidelity criteria affects stability and performance of control systems.

\appendix

\section{Proof of Theorem~\ref{nonstationary:theorem:alternative_expression_nonstationary}}\label{nonstationary:appendix:proof:theorem:alternative_expression_nrdf}

\par Let $s\leq{0}$, $\lambda_n\in\Psi_s^n$ and $\overrightarrow{Q}^*_{0,n}(\cdot|y^{-1},x^n)\in\overrightarrow{\cal Q}_{0,n}(D)$ be given. Then, using the fact that 
\begin{equation*}
\frac{1}{n+1}\sum_{t=0}^n\int_{{\cal X}^t\times{\cal Y}^t}\rho_t(T^tx^n,T^ty^n)({P}_{0,t}{\otimes}{\overrightarrow Q}_{0,t})\left(dx^t,dy_0^t|y^{-1}\right)\mu(dy^{-1})\leq{D}
\end{equation*}
we obtain
\begin{align}
&\frac{1}{n+1}\mathbb{I}_{0,n}({P}_{0,n},\overrightarrow{Q}_{0,n}) -sD +\frac{1}{n+1}\sum_{t=0}^n\int_{{\cal X}^t\times{\cal Y}^t}g_{t,n}(x^t,y^t)({P}_{0,t}{\otimes}{\overrightarrow Q}_{0,t})\left(dx^t,dy_0^t|y^{-1}\right)\otimes\mu(dy^{-1})\nonumber\\
&-\frac{1}{n+1}\sum_{t=0}^n\int_{{\cal X}^t\times{\cal Y}^{t-1}}\log\left(\lambda_t(x^t,y^{t-1})\right)({P}_{0,t}{\otimes}{\overrightarrow Q}_{0,t})\left(dx^t,dy_0^t|y^{-1}\right)\otimes\mu(dy^{-1})\nonumber\\
&\geq\frac{1}{n+1}\sum_{t=0}^n\int_{{\cal X}^t\times{\cal Y}^t}\log\left(\frac{Q^*_t(dy_t|y^{t-1},x^t)}{\Pi^{\overrightarrow{Q}^*}_{t}(dy_t|y^{t-1})}\right)({P}_{0,t}{\otimes}{\overrightarrow Q}_{0,t})\left(dx^t,dy_0^t|y^{-1}\right)\otimes\mu(dy^{-1})\nonumber\\
&\qquad-s\frac{1}{n+1}\sum_{t=0}^n\int_{{\cal X}^t\times{\cal Y}^t}\rho_t(T^tx^n,T^ty^n)({P}_{0,t}{\otimes}{\overrightarrow Q}_{0,t})\left(dx^t,dy_0^t|y^{-1}\right)\otimes\mu(dy^{-1})\nonumber\\
&\qquad+\frac{1}{n+1}\sum_{t=0}^n\int_{{\cal X}^t\times{\cal Y}^t}g_{t,n}(x^t,y^t)({P}_{0,t}{\otimes}{\overrightarrow Q}_{0,t})\left(dx^t,dy_0^t|y^{-1}\right)\otimes\mu(dy^{-1})\nonumber\\
&\qquad-\frac{1}{n+1}\sum_{t=0}^n\int_{{\cal X}^t\times{\cal Y}^t}\log\left(\lambda_t(x^t,y^{t-1})\right)({P}_{0,t}{\otimes}{\overrightarrow Q}_{0,t})\left(dx^t,dy_0^t|y^{-1}\right)\otimes\mu(dy^{-1})\nonumber\\
&=\frac{1}{n+1}\sum_{t=0}^n\int_{{\cal X}^{t-1}\times{\cal Y}^{t-1}}\Bigg\{\int_{{\cal X}_t\times{\cal Y}_t}\log\left(\frac{Q^*_t(dy_t|y^{t-1},x^t)e^{-s\rho_t(T^tx^n,T^ty^n)+g_{t,n}(x^t,y^t)}}{\Pi^{\overrightarrow{Q}^*}_{t}(dy_t|y^{t-1})\lambda_t(x^t,y^{t-1})}\right)\nonumber\\
&\qquad Q^*_t(dy_t|y^{t-1},x^t)\otimes{P}_t(dx_t|x^{t-1})\Bigg\}({P}_{0,t-1}{\otimes}{\overrightarrow Q}_{0,t-1})\left(dx^{t-1},dy_0^{t-1}|y^{-1}\right)\otimes\mu(dy^{-1})\nonumber\\
&\stackrel{(a)}\geq\frac{1}{n+1}\sum_{t=0}^n\int_{{\cal X}^{t-1}\times{\cal Y}^{t-1}}\Bigg\{\int_{{\cal X}_t\times{\cal Y}_t}\left(1-\frac{e^{s\rho_t(T^tx^n,T^ty^n)-g_{t,n}(x^t,y^t)}\Pi^{\overrightarrow{Q}^*}_{t}(dy_t|y^{t-1})\lambda_t(x^t,y^{t-1})}{Q^*_t(dy_t|y^{t-1},x^t)}\right)\nonumber\\
&\qquad\qquad {Q}^*_t(dy_t|y^{t-1},x^t)\otimes{P}_t(dx_t|x^{t-1})\Bigg\}({P}_{0,t-1}{\otimes}{\overrightarrow Q}_{0,t-1})\left(dx^{t-1},dy_0^{t-1}|y^{-1}\right)\mu(dy^{-1})\nonumber \\
&=\frac{1}{n+1}\sum_{t=0}^n\Bigg\{1-\int_{{\cal Y}^t}\Pi^{\overrightarrow{Q}^*}_{0,t}(dy_0^{t}|y^{-1})\otimes\mu(dy^{-1})\nonumber\\
&\left(\int_{{\cal X}^{t-1}}\int_{{\cal X}_t}{e^{s\rho_t(T^tx^n,T^ty^n)-g_{t,n}(x^t,y^t)}\lambda_t(x^t,y^{t-1})}P_t(dx_t|x^{t-1})\otimes{\bf P}^{\overrightarrow{Q}^*}(dx^{t-1}|y^{t-1})\right)\Bigg\}\nonumber\\
&\stackrel{(b)}\geq\frac{1}{n+1}\sum_{t=0}^n\left(1-\int_{{\cal Y}^t}\Pi^{\overrightarrow{Q}^*}_{0,t}(dy_0^t|y^{-1})\otimes\mu(dy^{-1})\right)=0\nonumber
\end{align}
where $(a)$ follows from the inequality $\log{x}\geq1-\frac{1}{x},~x>0$, and $(b)$ follows from (\ref{nonstationary:eq.39}).\\
Hence, we obtain
\begin{align}
\begin{split}
&R^{na}_{0,n}(D)\stackrel{(c)}\geq\\
&\sup_{s\leq{0}}\sup_{\lambda\in\Psi_s}\Big\{sD-\frac{1}{n+1}\sum_{t=0}^n\int_{{\cal X}^t\times{\cal Y}^t}g_{t,n}(x^t,y^t)({P}_{0,t}{\otimes}{\overrightarrow Q}_{0,t})\left(dx^{t},dy_0^{t}|y^{-1}\right)\otimes\mu(dy^{-1})\\
&+\frac{1}{n+1}\sum_{t=0}^n\int_{{\cal X}^t\times{\cal Y}^{t-1}}\log\left(\lambda_t(x^t,y^{t-1})\right){P}_t(dx_t|x^{t-1})({P}_{0,t-1}{\otimes}{\overrightarrow Q}_{0,t-1})\left(dx^{t-1},dy_0^{t-1}|y^{-1}\right)\otimes\mu(dy^{-1})\Big\}.\end{split}
\nonumber
\end{align}
However, equality in $(c)$ holds if 
\begin{align*}
\lambda_t(x^t,y^{t-1})\triangleq\left(\int_{{\cal Y}_t}e^{s\rho_t(T^tx^n,T^ty^n)-g_{t,n}(x^t,y^t)}\Pi^{\overrightarrow{Q}^*}_{t}(dy_t|y^{t-1})\right)^{-1},~\forall{t}\in\mathbb{N}_0^n.
\end{align*}
\par This completes the proof.

\section{Proof of Theorem~\ref{nonstationary:example:gaussian:solution_gaussian}}\label{nonstationary:theorem:proof:solution_gaussian}

{\bf (1)} The derivation is based on the fact that the feedback realization scheme of Fig.~\ref{nonstationary:communication_system} is generally an upper bound on the NRDF, $R_{0,n}^{na}(D)$, of the Gaussian process, and this realization gives (\ref{nonstationary:example:gaussian:equa.13}). The achievability of this upper bound is established by evaluating the lower bound in (\ref{nonstationary:example:gaussian:equa.13}) which is done recursively moving backward in time, utilizing the expression we obtained in Theorem~\ref{nonstationary:theorem:alternative_expression_nonstationary}.\\
{\it Upper Bound.} First, consider the realization of Fig.~\ref{nonstationary:communication_system}. Define $\{H_t:~{t\in\mathbb{N}_0^{n}}\}$ as in (\ref{nonstationary:example:gaussian:equation11mmm}). By Fig.~\ref{nonstationary:communication_system}, we obtain
\begin{align}
\tilde{K}_t&=E_t^{\T}H_tE_t\left(X_t-\mathbb{E}\left\{X_t|\sigma\{Y^{t-1}\}\right\}\right)+E_t^{\T}{\Phi}_tV^c_t=E_t^{\T}H_tE_tK_t+E_t^{\T}{\Phi}_tV^c_t,~{t\in\mathbb{N}_0^{n}}\label{nonstationary:example:gaussian:equation95ii}
\end{align}
where $\{V^c_t: {t\in\mathbb{N}_0^{n}}\}$ is a zero mean independent Gaussian process with covariance $\cov(V^c_t)=Q_t=\diag\{q_{t,1},\ldots,q_{t,p}\}$, and $\{{\Phi}_t:~{t\in\mathbb{N}_0^{n}}\}$ is to be determined. Next, we show that by letting ${\Phi}_t=\sqrt{H_t\Delta_t{Q}_t^{-1}}$, and $\Delta_t\triangleq\diag\{\delta_{t,1},\ldots,\delta_{t,p}\}$, then $\Pi_{t|t-1}=\mathbb{E}\left\{K_tK^{\T}_t\right\}$, and also $\frac{1}{n+1}\mathbb{E}\left\{\sum_{t=0}^n||X_t-Y_t||_{2}^2\right\}=\frac{1}{n+1}\mathbb{E}\left\{\sum_{t=0}^n||K_t-\tilde{K}_t||_{2}^2\right\}=D$. Clearly, by (\ref{nonstationary:example:gaussian:equation52}), (\ref{nonstationary:example:gaussian:eq.12i}), and (\ref{nonstationary:example:gaussian:equation95ii}), we obtain
\begin{align}
&\frac{1}{n+1}\sum_{t=0}^n\mathbb{E}\left\{(X_t-{Y}_t)^{\T}(X_t-{Y}_t)\right\}=\frac{1}{n+1}\sum_{t=0}^n{\trace}\left(\mathbb{E}\left\{(K_t-\tilde{K}_t)(K_t-\tilde{K}_t)^{\T}\right\}\right)\nonumber\\
&=\frac{1}{n+1}\sum_{t=0}^n{\trace}~\mathbb{E}\left\{(K_t-E_t^{\T}H_tE_tK_t-E_t^{\T}{\Phi}_tV^c_t)(K_t-E_t^{\T}H_tE_tK_t-E_t^{\T}{\Phi}_tV^c_t)^{\T}\right\}\nonumber\\
&=\frac{1}{n+1}\sum_{t=0}^n{\trace}\left\{E_t^{\T}\left((I-H_t)\diag(\lambda_{t,1},\ldots,\lambda_{t,p})(I-H_t)^{\T}+({\Phi}_tQ_t{\Phi}_t^{\T})\right)E_t\right\}\nonumber\\
&\stackrel{(a)}=\frac{1}{n+1}\sum_{t=0}^n{\trace}\left\{\diag(\delta_{t,1},\ldots,\delta_{t,p})\right\}=D , \nonumber
\end{align}
where $(a)$ holds by setting ${\Phi}_t$ as in (\ref{nonstationary:example:gaussian:equation11mmm}). By (\ref{nonstationary:example:gaussian:equation.2}), the NRDF can be written as follows:
\begin{align}
R_{0,n}^{na,K^n,\tilde{K}^n}(D)&\leq\frac{1}{n+1}\sum_{t=0}^n{I}(K_t;\tilde{K}_t|\tilde{K}^{t-1})\label{nonstationary:example:gaussian:equa.18-}\\
&=\frac{1}{n+1}\sum_{t=0}^n\left\{{H}(\tilde{K}_t|\tilde{K}^{t-1})-H(\tilde{K}_t|\tilde{K}^{t-1},K_t)\right\}\nonumber\\
&\stackrel{(b)}\leq\frac{1}{n+1}\sum_{t=0}^n\left\{{H}(\tilde{K}_t)-H(\tilde{K}_t|\tilde{K}^{t-1},K_t)\right\}\nonumber\\
&\stackrel{(c)}\leq\frac{1}{n+1}\sum_{t=0}^n\left\{{H}(\tilde{K}_t)-H(\tilde{K}_t|K_t)\right\}\nonumber\\
&\stackrel{(d)}\leq\sum_{t=0}^n\left\{{H}(\tilde{K}_t)-H(E_t^{\T}{\Phi}_tV^c_t)\right\} , \label{nonstationary:example:gaussian:equa.18}
\end{align}
where $(b)$ follows from the fact that conditioning reduces entropy (see also \cite[Lemma V.1, Remark V.2]{stavrou-kourtellaris-charalambous2016ieeetit}), $(c)$ follows again from the fact that $\tilde{K}_t=E_t^{\T}H_tE_tK_t+E_t^{\T}{\Phi}_tV^c_t$ is a memoryless Gaussian channel, and $(d)$ follows from the orthogonality of $K_t$ and $V_t^c$. Actually, by \cite[Lemma V.1, Remark V.2]{stavrou-kourtellaris-charalambous2016ieeetit}, it can be shown that the inequalities $(b)$, $(c)$, $(d)$ are equalities.\\
\noi Next, we compute the entropies appearing in (\ref{nonstationary:example:gaussian:equa.18}) from the covariances of the corresponding processes. The covariance of the Gaussian zero mean term $E_t^{\T}{\Phi}_tV^c_t,~{t\in\mathbb{N}_0^{n}}$, is given by
\begin{align}
\mathbb{E}\left\{(E_t^{\T}{\Phi}_tV^c_t)(E_t{\Phi}_tV_t^{c})^{\T}\right\}&=E_t^{\T}{\Phi}_t\mathbb{E}\{V^c_tV_t^{c,{\T}}\}{\Phi}_t^{\T}E_t=E_t^{\T}{\Phi}_tQ_t{\Phi}_t^{\T}E_t\nonumber\\
&=E_t^{\T}H_t\Delta_tE_t=E_t^{\T}\diag\{\eta_{t,1}\delta_{t,1},\ldots,\eta_{t,p}\delta_{t,p}\}E_t,~{t\in\mathbb{N}_0^{n}}.\label{nonstationary:example:gaussian:equa.19}
\end{align}
The covariance of $\tilde{K}_t,~{t\in\mathbb{N}_0^{n}}$, is given by
\begin{align}
\mathbb{E}\left\{\tilde{K}_t\tilde{K}^{\T}_t\right\}&=\mathbb{E}\left\{(E_t^{\T}H_tE_tK_t+E_t^{\T}{\Phi}_tV^c_t)(E_t^{\T}H_tE_tK_t+E_t^{\T}{\Phi}_tV^c_t)^{\T}\right\}\nonumber\\
&=E_t^{\T}\big(\diag\{\eta^2_{t,1}\lambda_{t,1},\ldots,\eta^2_{t,p}\lambda_{t,p}\}+\diag\{\eta_{t,1}\delta_{t,1},\ldots,\eta_{t,p}\delta_{t,p}\}\big)E_t\nonumber\\
&=E_t^{\T}\diag\{\lambda_{t,1}-\delta_{t,1},\ldots,\lambda_{t,p}-\delta_{t,p}\}E_t,~{t\in\mathbb{N}_0^{n}}.\label{nonstationary:example:gaussian:equa.20}
\end{align}
Using (\ref{nonstationary:example:gaussian:equa.20}) we obtain the first term of (\ref{nonstationary:example:gaussian:equa.18}) as follows\footnote{Note that $(\cdot)^{+}\triangleq\max\{0,\cdot\}$.}
\begin{align}
\sum_{t=0}^n{H}(\tilde{K}_t)=\frac{1}{2}\sum_{t=0}^n\sum_{i=1}^p\log\left\{\left(2{\pi}e\right)\left(\lambda_{t,i}-\delta_{t,i}\right)^{+}\right\}.\label{nonstationary:example:gaussian:equa.21}
\end{align}
Also, by (\ref{nonstationary:example:gaussian:equa.19}), we obtain the second term in (\ref{nonstationary:example:gaussian:equa.18}) as follows.
\begin{align}
\sum_{t=0}^n{H}(E_t^{\T}{\Phi}_tV^c_t)=\frac{1}{2}\sum_{t=0}^n\sum_{i=1}^p\log\left\{\left(2{\pi}e\right)\left(\eta_{t,i}\delta_{t,i}\right)\right\}.\label{nonstationary:example:gaussian:equa.22}
\end{align}
This problem can be cast into the following convex optimization problem
\begin{align}
\min_{\frac{1}{n+1}\sum_{t=0}^n\sum_{i=1}^p\delta_{\infty,i}=D}\frac{1}{n+1}\sum_{t=0}^n\sum_{i=1}^p\max\left\{0,\frac{1}{2}\log\left(\frac{\lambda_{t,i}}{\delta_{t,i}}\right)\right\}.\label{optimization problem:eq.1}
\end{align}
Since this is a convex optimization problem, we use Lagrange multipliers to construct the following augmented functional
\begin{align}
J(D)=\frac{1}{2}\sum_{t=0}^n\sum_{i=1}^p\log\left(\frac{\lambda_{t,i}}{\delta_{t,i}}\right)-s\frac{1}{n+1}\sum_{t=0}^n\sum_{i=1}^p\delta_{t,i},~~~s\leq{0}.\label{lagrangian_functional}
\end{align}
Differentiating with respect to $\delta_{t,i}$ and setting equal to zero, we obtain
\begin{align}
\frac{\partial{J}}{\partial{\delta_{t,i}}}=-\frac{1}{2\delta_{t,i}}-s=0\Longrightarrow{s}=-\frac{1}{2\delta_{t,i}}~~~\mbox{or}~~~\delta_{t,i}=\hat{\xi},~~\hat{\xi}\geq{0}.\label{lagrangian_functional:solution:eq.1}
\end{align}
Evidently, the optimal information allocation to the various descriptions results in an equal distortion for the components of the time-invariant multidimensional Gauss-Markov process. This is feasible if the constant $\hat{\xi}$ in \eqref{lagrangian_functional:solution:eq.1} is less than $\lambda_{t,i}$ $\forall{t,i}$. As the total distortion level increases, the constant $\hat{\xi}$ also increases until it exceeds $\lambda_{t,i}$ for some $t,i$. If we increase the total distortion, we must use the Karush-Kuhn-Tucker (KKT) conditions \cite{boyd-vandenberghe2004} to find the minimum in the convex optimization problem \eqref{optimization problem:eq.1}. By applying KKT conditions we obtain
\begin{align}
\frac{\partial{J}}{\partial{\delta_{t,i}}}=-\frac{1}{2\delta_{t,i}}-s,~~~s\leq{0}\label{lagrangian_functional:solution:eq.3}
\end{align}
where $s$ is chosen so that
\begin{align}
\frac{\partial{J}}{\partial{\delta_{t,i}}}=\left\{ \begin{array}{ll}0 & \mbox{if} \quad \delta_{t,i}\leq\lambda_{t,i} \\
\leq{0} &  \mbox{if}\quad\delta_{t,i}>\lambda_{t,i} \end{array} \right..\label{lagrangian_functional:solution:eq.4}
\end{align}
It is easy to verify that the solution of KKT conditions yields
\begin{align}
\delta_{t,i} \triangleq\left\{ \begin{array}{ll} \xi & \mbox{if} \quad \xi\leq\lambda_{t,i} \\
\lambda_{t,i} &  \mbox{if}\quad\xi>\lambda_{t,i} \end{array} \right.,~\forall{t,i}\label{lagrangian_functional:solution:eq.5}
\end{align}
where $\xi$ is chosen such that $\frac{1}{n+1}\sum_{t=0}^n\sum_{i=1}^p\delta_{t,i}=D$ and
$\delta_{t,i}=\min\{\xi,\lambda_{t,i}\}$.
\par Using (\ref{nonstationary:example:gaussian:equa.21}) and (\ref{nonstationary:example:gaussian:equa.22}) in (\ref{nonstationary:example:gaussian:equa.18}) we have the following upper bound
\begin{align}
\begin{split}
R_{0,n}^{na,K^n,\tilde{K}^n}(D)&\leq\frac{1}{n+1}\sum_{t=0}^n{I}(K_t;\tilde{K}_t|,\tilde{K}^{t-1})\\
&\leq\frac{1}{2}\frac{1}{n+1}\sum_{t=0}^n\sum_{i=1}^p\log\left\{\frac{\left(\lambda_{t,i}-\delta_{t,i}\right)^{+}}{\eta_{t,i}\delta_{t,i}}\right\}
=\frac{1}{2}\frac{1}{n+1}\sum_{t=0}^n\sum_{i=1}^p\log\left(\frac{\lambda_{t,i}}{\delta_{t,i}}\right) ,
\end{split}\label{nonstationary:example:gaussian:equation101}
\end{align}
where $\delta_{t,i}=\min\{\xi,\lambda_{t,i}\},~{t\in\mathbb{N}_0^{n}}, i=1,\ldots,p$, and $\frac{1}{n+1}\sum_{t=0}^n\sum_{i=1}^p\delta_{t,i}=D$. Note that if $\delta_{t,i}=\lambda_{t,i}$, for~${t\in\mathbb{N}_0^{n}}$ and $i=1,\ldots,p$, then no data are estimated.
\vspace*{0.2cm}\\
\noi{\it Lower Bound.} Here, we apply Theorem~\ref{nonstationary:theorem:alternative_expression_nonstationary} recursively, to obtain a lower bound for the NRDF, ${R}_{0,n}^{na}(D)=R_{0,n}^{na, K^n,\tilde{K}^n}(D)$, which is precisely (\ref{nonstationary:example:gaussian:equa.13}).\\
Let $\bar{p}(\cdot|,\cdot)$ and $\bar{p}(\cdot)$ denote the conditional and unconditional densities, respectively. Using the property of $\{\lambda_t(\cdot,\cdot,):~t=0,\ldots,n\}$ corresponding to the fact that $\lambda_t(k^t,\tilde{k}^{t-1})\equiv\lambda_t(k_t,\tilde{k}^{t-1})$, ~$t=0,\ldots,n$ and by Theorem~\ref{nonstationary:theorem:alternative_expression_nonstationary}, an alternative expression for the NRDF, $R_{0,n}^{na, K^n,\tilde{K}^n}(D)$ is the following.
\begin{equation}
R_{0,n}^{na,K^n,\tilde{K}^n}(D)=\sup_{s\leq0}\sup_{\{\lambda_t(k_t,\tilde{k}^{t-1})\in{\Psi}_s^t:~t\in\mathbb{N}_0^n\}}\left\{\mbox{term-(0)+\ldots+term-(n-1)+term-(n)}\right\}\label{nonstationary:example:gaussian:equation108}
\end{equation}
where
\begin{align*}
\mbox{term-(0)}&\equiv-\frac{1}{n+1}\int_{{\cal K}_0}\Big(\int_{\tilde{{\cal K}}_{0}}g_{0,n}(\tilde{k}_0)\bar{p}(\tilde{k}_0|k_0)d\tilde{k}_0\Big)\bar{p}(k_0|)dk_0+\frac{1}{n+1}\int_{{\cal K}_0}\log\Big(\lambda_0(k_0)\Big)\bar{p}(k_0)dk_0\\
\mbox{term-(1)}&\equiv-\frac{1}{n+1}\int_{{\cal K}_1\times\tilde{{\cal K}}_0}\Big(\int_{\tilde{{\cal K}}_{1}}g_{1,n}(\tilde{k}^1)\bar{p}(\tilde{k}_1|\tilde{k}_0,k_1)d\tilde{k}_1\Big)\bar{p}(k_1,\tilde{k}_0)dk_1d\tilde{k}_0\nonumber\\
&+\frac{1}{n+1}\int_{{\cal K}_1\times\tilde{{\cal K}}_0}\log\Big(\lambda_0(k_1,\tilde{k}_0)\Big)\bar{p}(k_1,\tilde{k}_0)dk_1d\tilde{k}_0\\
&\vdots\\
\mbox{term-(n-2)}&\equiv-\frac{1}{n+1}\int_{{\cal K}_{n-2}\times\tilde{\cal K}^{n-3}}\Big(\int_{\tilde{\cal K}_{n-2}}g_{n-2,n}(\tilde{k}^{n-2})\bar{p}(\tilde{k}_{n-2}|\tilde{k}^{n-3},k_{n-2})d\tilde{k}_{n-2}\Big)\bar{p}(k_{n-2},\tilde{k}^{n-3})\nonumber\\
&dk_{n-2}d\tilde{k}^{n-3}+\frac{1}{n+1}\int_{{\cal K}_{n-2}\times\tilde{\cal K}^{n-3}}\log\Big(\lambda_{n-2}(k_{n-2},\tilde{k}^{n-3})\Big)\bar{p}(k_{n-2},\tilde{k}^{n-3})dk_{n-2}d\tilde{k}^{n-3}\\
\mbox{term-(n-1)}&\equiv-\frac{1}{n+1}\int_{{\cal K}_{n-1}\times\tilde{{\cal K}}^{n-2}}\Big(\int_{\tilde{{\cal K}}_{n-1}}g_{n-1,n}(\tilde{k}^{n-1})\bar{p}(\tilde{k}_{n-1}|\tilde{k}^{n-2},k_{n-1})d\tilde{k}_{n-1}\Big)\bar{p}(k_{n-1},\tilde{k}^{n-2})\\
&dk_{n-1}d\tilde{k}^{n-2}+\frac{1}{n+1}\int_{{\cal K}_{n-1}\times\tilde{{\cal K}}^{n-2}}\log\Big(\lambda_{n-1}(k_{n-1},\tilde{k}^{n-2})\Big)\bar{p}(k_{n-1},\tilde{k}^{n-2})dk_{n-1}d\tilde{k}^{n-2}
\end{align*}
\begin{align*}
\mbox{term-(n)}\equiv{s}D+\frac{1}{n+1}\int_{{{\cal K}}_n\times\tilde{{\cal K}}^{n-1}}\log\Big(\lambda_n(k_n,\tilde{k}^{n-1})\Big)\bar{p}(k_n,\tilde{k}^{n-1})dk_nd\tilde{k}^{n-1}
\end{align*}
and
\begin{align}
\Psi^t_s\triangleq\Big\{&\lambda_t(k_t,\tilde{k}^{t-1})\geq{0}:
~\int_{{\cal K}_{t}}{e}^{s||k_t-\tilde{k}_t||_{2}^2-g_{t,n}(\tilde{k}^t)}{\lambda_t(k_t,\tilde{k}^{t-1})}\bar{p}(k_t|\tilde{k}^{t-1})dk_t\leq{1}\Big\},~t\in\mathbb{N}_{0}^n,\label{nonstationary:example:gaussian:equation100aaa}
\end{align}
\begin{align}
g_{n,n}(\tilde{k}^n)&=0,\nonumber\\
g_{t,n}(\tilde{k}^t)&=-\int_{{\cal K}_{t+1}}\log\Big({\lambda_{t+1}(k_{t+1},\tilde{k}^{t})}\Big)^{-1}\bar{p}(k_{t+1}|,\tilde{k}^t)dk_{t+1},~t\in\mathbb{N}_{0}^{n-1}.\label{nonstationary:example:gaussian:equation100aa}
\end{align}
Clearly, if $g_{t,n}(\tilde{k}^t)=\bar{g}_{t,n}(\tilde{k}^{t-1})$, i.e.,  it is independent of $(\tilde{k}_t)$, for ${t\in\mathbb{N}_0^{n-1}}$, then by Theorem~\ref{nonstationary:theorem:solution_rdf}, the RHS terms in (\ref{nonstationary:example:gaussian:equation108}) involving $g_{t,n}(\cdot,\cdot),~t\in\mathbb{N}_{0}^{n-1}$, will not appear (because the optimal reproduction distribution will not involve such terms).\\
\noi Since $g_{n,n}(\cdot,\cdot)=0$, by (\ref{nonstationary:example:gaussian:equation100aaa}), (\ref{nonstationary:example:gaussian:equation100aa}), $\lambda_n(k_n,\tilde{k}^{n-1})$ determines $g_{n-1,n}(\cdot,\cdot)$, $\lambda_{n-1}(\cdot,\cdot)$ determines $g_{n-2,n}(\cdot,\cdot)$ and so on, and the RHS of (\ref{nonstationary:example:gaussian:equation108}) involves supremum over $\{\lambda_t(\cdot,\cdot):~{t\in\mathbb{N}_0^{n}}\}$, then 
any choice of $\{\lambda_t(\cdot,\cdot):~{t\in\mathbb{N}_0^{n}}\}$ gives a lower bound.\\
\noi The main idea, implemented below, uses the property of distortion function, and the source distribution, to show that $\{\lambda_t(\cdot,\cdot):~{t\in\mathbb{N}_0^{n}}\}$ can be chosen so that $g_{t,n}(\tilde{k}^t)=\bar{g}_{t,n}(\tilde{k}^{t-1}),~{t\in\mathbb{N}_0^{n-1}}$, giving a lower bound which is achievable, and that the optimal reproduction distribution is of the form 
\begin{align*}
\bar{p}(\tilde{k}_t|,\tilde{k}^{t-1},k_t)=\frac{e^{s||{k}_t-\tilde{k}_t||_{2}^2}\bar{p}(\tilde{k}_t|\tilde{k}^{t-1})}{\int_{{\cal \tilde{K}}_t}e^{s||{k}_t-\tilde{k}_t||_{2}^2}\bar{p}(\tilde{k}_t|\tilde{k}^{t-1})}.
\end{align*}

\paragraph*{\bf {Step ${\bf t=n}$}} The set $\Psi^n_s$ is defined as follows:
\begin{align}
\Psi^n_s\triangleq\Big\{&\lambda_n(k_n,\tilde{k}^{n-1})\geq{0}:
~\int_{{\cal K}_{n}}{e}^{s||{k}_n-\tilde{{k}}_n||_{2}^2}{\lambda_n(k_n,\tilde{k}^{n-1})}\bar{p}(k_n|\tilde{k}^{n-1})dk_n\leq{1}\Big\}  \label{nonstationary:example:gaussian:equa.24} ,
\end{align}
where $\bar{p}(k_n|\tilde{k}^{n-1})$ denotes the conditional density of $k_n$ given $(\tilde{k}^{n-1})$. Take $\lambda_n(k_n,\tilde{k}^{n-1})\in\Psi^n_s$ such that
\begin{align}
\lambda_n(k_n,\tilde{k}^{n-1})=\frac{\alpha_n}{\bar{p}(k_n|\tilde{k}^{n-1})}\label{nonstationary:example:gaussian:equa.23}
\end{align}
for some $\alpha_n$ not depending on $k_n$, and substitute (\ref{nonstationary:example:gaussian:equa.23}) into the integral inequality in (\ref{nonstationary:example:gaussian:equa.24}) to obtain
\begin{align*}
{\alpha_n}\int_{{\cal K}_n}{e}^{s||{k}_n-\tilde{{k}}_n||_{2}^2}dk_n\leq{1}.
\end{align*}
By change of variable of integration then
\begin{align}
{\alpha_n}\int_{{-\infty}}^{{\infty}} {e}^{s||z_n||^2_{2}}dz_n=\alpha_n\sqrt{\left(-\frac{\pi}{s}\right)^{p}}=\alpha_n\left(-\frac{\pi}{s}\right)^{\frac{p}{2}}\leq{1}\label{nonstationary:example:gaussian:equa.25} ,
\end{align}
where $``s"$ is the non-positive Lagrange multiplier.\\
\noi Moreover, $\alpha_n$ is chosen so that the inequality of (\ref{nonstationary:example:gaussian:equa.25}) holds with equality, giving 
\begin{align}
\alpha_n=\frac{1}{\int{e}^{s||z_n||_{2}^2}dz_n}=\left(-\frac{s}{\pi}\right)^{\frac{p}{2}},~\lambda_n(k_n,\tilde{k}^{n-1})=\frac{(-\frac{s}{\pi})^{p/2}}{\bar{p}_n(k_n|\tilde{k}^{n-1})}\label{nonstationary:example:gaussian:equa.25a}.
\end{align}
Substituting (\ref{nonstationary:example:gaussian:equa.25a}) into the term-(n) of (\ref{nonstationary:example:gaussian:equation108}) gives
\begin{align}
\mbox{\underline{term-(n)}}&=sD+\frac{1}{n+1}\log{\alpha_n}-\frac{1}{n+1}\int_{{\cal K}_n\times\tilde{\cal K}^{n-1}}\log\Big({\bar{p}(k_n|\tilde{k}^{n-1})}\Big)\bar{p}(k_n,\tilde{k}^{n-1})dk_nd\tilde{k}^{n-1}\nonumber\\
&=sD+\frac{1}{n+1}\log\left(-\frac{s}{\pi}\right)^{\frac{p}{2}}+\frac{1}{n+1}H(K_n|\tilde{K}^{n-1}).\label{nonstationary:example:gaussian:equation110}
\end{align}
The choice of $\lambda_n(\cdot,\cdot)$ given by (\ref{nonstationary:example:gaussian:equa.25a}) determines $g_{n-1,n}(\cdot)$ given by 
\begin{align}
g_{n-1,n}(\tilde{k}^{n-1})&=-\int_{{\cal K}_n}\bar{p}(dk_n|\tilde{k}^{n-1})\log\Big({\lambda_n(k_n,\tilde{k}^{n-1}})\Big)^{-1}\nonumber\\
&\stackrel{(a)}=-\int_{{\cal K}_n}\bar{p}(dk_n|\tilde{k}^{n-1})\log\Big(\frac{\bar{p}(k_n|\tilde{k}^{n-1})}{\alpha_n}\Big)\nonumber\\
&=\log{\alpha_n}+H(K_n|\tilde{K}^{n-1}=\tilde{k}^{n-1}),~\alpha_n=\left(-\frac{s}{\pi} \right)^{\frac{p}{2}}\nonumber\\
&\stackrel{(b)}\leq\log{\alpha_n}+H(K_n|\tilde{K}^{n-2}=\tilde{k}^{n-2})\nonumber\\
&=\log\left(-\frac{s}{\pi}\right)^{\frac{p}{2}}+H(K_n|\tilde{K}^{n-2}=\tilde{k}^{n-2})\equiv\bar{g}_{n-1,n}(\tilde{k}^{n-2})\label{nonstationary:example:gaussian:equation113}
\end{align}
where $(a)$ follows from the fact that $\Big(\lambda_n(k_n,\tilde{k}^{n-1})\Big)^{-1}=\frac{\bar{p}(k_n|\tilde{k}^{n-1})}{\alpha_n}$, and $(b)$ from the fact that conditioning reduces entropy.\\
When the upper bound in (\ref{nonstationary:example:gaussian:equation113})  is substituted into the second expression of term-(n-1) of (\ref{nonstationary:example:gaussian:equation108}) involving $g_{n-1,n}(\cdot)$, it gives
\begin{align*}
&-\frac{1}{n+1}\int_{{\cal K}_{n-1}\times\tilde{\cal K}^{n-2}}\Big(\int_{\tilde{\cal K}_{n-1}}g_{n-1,n}(\tilde{k}^{n-1})\bar{p}(\tilde{k}_{n-1}|\tilde{k}^{n-2},k_{n-1})d\tilde{k}_{n-1}\Big)\bar{p}(k_{n-1},\tilde{k}^{n-2})dk_{n-1}d\tilde{k}^{n-2}\nonumber\\
&\geq-\frac{1}{n+1}\int_{{\cal K}_{n-1}\times\tilde{\cal K}^{n-2}}\Big(\int_{\tilde{\cal K}_{n-1}}\bar{g}_{n-1,n}(\tilde{k}^{n-2})\bar{p}(\tilde{k}_{n-1}|\tilde{k}^{n-2},k_{n-1})d\tilde{k}_{n-1}\Big)\bar{p}(k_{n-1},\tilde{k}^{n-2})dk_{n-1}d\tilde{k}^{n-2}.
\end{align*}

\paragraph*{\bf{Step ${\bf t=n-1}$}} The set $\Psi^{n-1}_s$ is defined as follows (using $g_{n-1,n}(\tilde{k}^{n-1})\equiv\bar{g}_{n-1,n}(\tilde{k}^{n-2})$ given by (\ref{nonstationary:example:gaussian:equation113}) obtained in step $t=n$) 
\begin{align}
\Psi^{n-1}_s\triangleq&\Big\{\lambda_{n-1}(k_{n-1},\tilde{k}^{n-2})\geq{0}:\nonumber\\
&\int_{{\cal K}_{n-1}}{e}^{s||k_{n-1}-\tilde{k}_{n-1}||_{2}^2-\bar{g}_{n-1,n}(\tilde{k}^{n-2})}{\lambda_{n-1}(k_{n-1},\tilde{k}^{n-2})}\bar{p}(k_{n-1}|\tilde{k}^{n-2})dk_{n-1}\leq{1}\Big\}.\label{nonstationary:example:gaussian:equation112}
\end{align}
Take $\lambda_{n-1}(k_{n-1},\tilde{k}^{n-2})\in\Psi^{n-1}_s$ such that
\begin{align}
\lambda_{n-1}(k_{n-1},\tilde{k}^{n-2})=\frac{\alpha_{n-1}(\tilde{k}^{n-2})}{\bar{p}(k_{n-1}|\tilde{k}^{n-2})}\label{nonstationary:example:gaussian:equation114}
\end{align}
for some $\alpha_{n-1}(\tilde{k}^{n-2})$ not depending on $k_{n-1}$, and substitute (\ref{nonstationary:example:gaussian:equation114}) into the integral inequality in (\ref{nonstationary:example:gaussian:equation112}) to obtain
\begin{align*}
{\alpha_{n-1}}(\tilde{k}^{n-2}){e}^{-\bar{g}_{n-1,n}(\tilde{k}^{n-2})}\int_{{\cal K}_{n-1}}{e}^{s||k_{n-1}-\tilde{k
}_{n-1}||_{2}^2}dk_{n-1}\leq{1}.
\end{align*}
By change of variable of integration then
\begin{align*}
{\alpha_{n-1}}(\tilde{k}^{n-2}){e}^{-\bar{g}_{n-1,n}(\tilde{k}^{n-2})}&\int_{-\infty}^{\infty}{e}^{s||z_{n-1}||_{2}^2}dz_{n-1}={\alpha_{n-1}}(\tilde{k}^{n-2}){e}^{-\bar{g}_{n-1,n}(\tilde{k}^{n-2})}\left(-\frac{\pi}{s}\right)^{\frac{p}{2}}\leq{1}.
\end{align*}
Hence,
\begin{align}
{\alpha_{n-1}}(\tilde{k}^{n-2})\left(-\frac{\pi}{s}\right)^{\frac{p}{2}}\leq{e}^{\bar{g}_{n-1,n}(\tilde{k}^{n-2})}.\label{nonstationary:example:gaussian:equation115}
\end{align}
\noi Moreover, $\alpha_{n-1}(\cdot)$ is chosen so that the inequality in (\ref{nonstationary:example:gaussian:equation115}) holds with equality, giving 
\begin{align}
\alpha_{n-1}(\tilde{k}^{n-2})=\frac{{e}^{\bar{g}_{n-1,n}(\tilde{k}^{n-2})}}{\left(-\frac{\pi}{s}\right)^{\frac{p}{2}}}&={e}^{\log{\alpha_n}+H(K_n|\tilde{K}^{n-2}=\tilde{k}^{n-2})}\left(-\frac{s}{\pi}\right)^{\frac{p}{2}}\nonumber\\
&\stackrel{(c)}{=}\left(-\frac{s}{\pi}\right)^{p}e^{H(K_n|\tilde{K}^{n-2}=\tilde{k}^{n-2})},\label{nonstationary:example:gaussian:equation115a}
\end{align}
where $(c)$ holds due to \eqref{nonstationary:example:gaussian:equation113}. Therefore, (\ref{nonstationary:example:gaussian:equation114}) is given by
\begin{align}
\lambda_{n-1}(k_{n-1},\tilde{k}^{n-1})&=\frac{\left(-\frac{s}{\pi}\right)^{p} e^{H(K_n|\tilde{K}^{n-2}=\tilde{k}^{n-2})}}{\bar{p}(k_{n-1}|\tilde{k}^{n-2})}.\label{nonstationary:example:gaussian:equation115b}
\end{align}
Substituting (\ref{nonstationary:example:gaussian:equation115b}) into the term-(n-1) of (\ref{nonstationary:example:gaussian:equation108}) gives
\begin{align}
\mbox{\underline{term-(n-1)}}\stackrel{(d)}\geq&-\frac{1}{n+1}\int_{{\cal K}_{n-1}\times\tilde{\cal K}^{n-2}}\Big(\int_{\tilde{\cal K}_{n-1}}\bar{g}_{n-1,n}(\tilde{k}^{n-2})\bar{p}(\tilde{k}_{n-1}|\tilde{k}^{n-2},k_{n-1})d\tilde{k}_{n-1}\Big)\nonumber\\
&\times\bar{p}(k_{n-1},\tilde{k}^{n-2}|)dk_{n-1}d\tilde{k}^{n-2}\nonumber\\
&+\frac{1}{n+1}\int_{{\cal K}_{n-1}\times\tilde{\cal K}^{n-2}}\log\Big(\lambda_{n-1}(k_{n-1},\tilde{k}^{n-2})\Big)\bar{p}(k_{n-1},\tilde{k}^{n-2})dk_{n-1}d\tilde{k}^{n-2}\label{nonstationary:example:gaussian:equation116a}\\
&\stackrel{(e)}=-\frac{1}{n+1}\log\left(-\frac{s}{\pi}\right)^{\frac{p}{2}}-\frac{1}{n+1}H(K_n|\tilde{K}^{n-2})\nonumber\\
&+\frac{1}{n+1}\int_{{\cal K}_{n-1}\times\tilde{\cal K}^{n-2}}\log\Big(\frac{\alpha_{n-1}(\tilde{k}^{n-2})}{\bar{p}(k_{n-1}|\tilde{k}^{n-2})}\Big)\bar{p}(k_{n-1},\tilde{k}^{n-2})dk_{n-1}d\tilde{k}^{n-2}\nonumber\\
&=-\frac{1}{n+1}\log\left(-\frac{s}{\pi}\right)^{\frac{p}{2}}-\frac{1}{n+1}H(K_n|\tilde{K}^{n-2})+\frac{1}{n+1}\log\left(-\frac{s}{\pi}\right)^{p}\nonumber\\
&+\frac{1}{n+1}H(K_n|\tilde{K}^{n-2})+\frac{1}{n+1}H(K_{n-1}|\tilde{K}^{n-2})\nonumber \\
&=\frac{1}{n+1}\log\left(-\frac{s}{\pi}\right)^{\frac{p}{2}}+\frac{1}{n+1}H(K_{n-1}|\tilde{K}^{n-2})\label{nonstationary:example:gaussian:equation116b} ,
\end{align}
where $(d)$ follows from the fact that $g_{n-1,n}(\tilde{k}^{n-1})\leq\bar{g}_{n-1,n}(\tilde{k}^{n-2})$ (see (\ref{nonstationary:example:gaussian:equation113})) and $(e)$ follows by substituting (\ref{nonstationary:example:gaussian:equation113}) and (\ref{nonstationary:example:gaussian:equation114}) into the the second and third expression of (\ref{nonstationary:example:gaussian:equation116a}), respectively.\\

\noi The choice of $\lambda_{n-1}(\cdot,\cdot)$ (given by (\ref{nonstationary:example:gaussian:equation115b})) determines $g_{n-2,n}(\cdot)$ given by
\begin{align}
\begin{split}
g_{n-2,n}(\tilde{k}^{n-2})&=-\int_{{\cal K}_{n-1}}\bar{p}(k_{n-1}|,\tilde{k}^{n-2})\log\Big({\lambda_{n-1}(k_{n-1},\tilde{k}^{n-2}})\Big)^{-1}\\
\stackrel{(f)}=-\int_{{\cal K}_{n-1}}\bar{p}(k_{n-1}|&\tilde{k}^{n-2})\log\Big(\frac{\bar{p}(k_{n-1}|\tilde{k}^{n-2})}{\alpha_{n-1}(\tilde{k}^{n-2})}\Big),~\alpha_{n-1}(\tilde{k}^{n-2})=\left(-\frac{s}{\pi}\right)^{p} e^{H(K_n|\tilde{K}^{n-2}=\tilde{k}^{n-2})}\\
&=\log\big(\alpha_{n-1}(\tilde{k}^{n-2})\big)-\int_{{\cal K}_{n-1}}\bar{p}(k_{n-1}|\tilde{k}^{n-2})\log\Big(\bar{p}(k_{n-1}|\tilde{k}^{n-2})\Big)\\
&=\log\left(-\frac{s}{\pi}\right)^{p}+H(K_n|\tilde{K}^{n-2}=\tilde{k}^{n-2})+H(K_{n-1}|\tilde{K}^{n-2}=\tilde{k}^{n-2})\\
&\stackrel{(g)}\leq\log\left(-\frac{s}{\pi}\right)^{p}+H(K_n|,\tilde{K}^{n-3}=\tilde{k}^{n-3})+H(K_{n-1}|\tilde{K}^{n-3}=\tilde{k}^{n-3})\\
&\equiv\bar{g}_{n-2,n}(\tilde{k}^{n-3}) ,
\end{split}\label{nonstationary:example:gaussian:equation116c}
\end{align}
where $(f)$ follows from the fact that $\Big(\lambda_{n-1}(k_{n-1},\tilde{k}^{n-2})\Big)^{-1}=\frac{\bar{p}(k_{n-1}|\tilde{k}^{n-2})}{\alpha_{n-1}(\tilde{k}^{n-2})}$, and $(g)$ follows from the fact that conditioning reduces entropy.\\
When the upper bound in (\ref{nonstationary:example:gaussian:equation116c})  is substituted into the second expression of term-(n-2) of (\ref{nonstationary:example:gaussian:equation108}) involving $g_{n-2,n}(\cdot)$, it gives
\begin{align*}
&-\frac{1}{n+1}\int_{{\cal K}_{n-2}\times\tilde{\cal K}^{n-3}}\Big(\int_{\tilde{\cal K}_{n-2}}g_{n-2,n}(\tilde{k}^{n-2})\bar{p}(\tilde{k}_{n-2}|\tilde{k}^{n-3},k_{n-2})d\tilde{k}_{n-2}\Big)\bar{p}(k_{n-2},\tilde{k}^{n-3})dk_{n-2}d\tilde{k}^{n-3}\nonumber\\
&\geq-\frac{1}{n+1}\int_{{\cal K}_{n-2}\times\tilde{\cal K}^{n-3}}\Big(\int_{\tilde{\cal K}_{n-2}}\bar{g}_{n-2,n}(\tilde{k}^{n-3})\bar{p}(\tilde{k}_{n-2}|\tilde{k}^{n-3},k_{n-2})d\tilde{k}_{n-2}\Big)\bar{p}(k_{n-2},\tilde{k}^{n-3})dk_{n-2}d\tilde{k}^{n-3}.
\end{align*}
\paragraph*{\bf {Step ${\bf t=n-2}$}} The set $\Psi^{n-2}_s$ is defined as follows (using ${g}_{n-2,n}(\tilde{k}^{n-2})\equiv\bar{g}_{n-2,n}(\tilde{k}^{n-3})$ given by (\ref{nonstationary:example:gaussian:equation116c}) obtained in step $t=n-1$). 
\begin{align}
\begin{split}
\Psi^{n-2}_s\triangleq&\Big\{\lambda_{n-2}(k_{n-2},\tilde{k}^{n-3})\geq{0}:\\
&\int_{{\cal K}_{n-2}}{e}^{s||k_{n-2}-\tilde{k}_{n-2}||_{2}^2-\bar{g}_{n-2,n}(\tilde{k}^{n-3})}{\lambda_{n-2}(k_{n-2},\tilde{k}^{n-3})}\bar{p}(k_{n-2}|\tilde{k}^{n-3})dk_{n-2}\leq{1}\Big\}.\end{split}\label{nonstationary:example:gaussian:equation118}
\end{align}
\noi Take $\lambda_{n-2}(k_{n-2},\tilde{k}^{n-3})\in\Psi^{n-2}_s$ such that
\begin{align}
\lambda_{n-2}(k_{n-2},\tilde{k}^{n-3})=\frac{\alpha_{n-2}(\tilde{k}^{n-3})}{\bar{p}(k_{n-2}|\tilde{k}^{n-3})}\label{nonstationary:example:gaussian:equation120}
\end{align}
for some $\alpha_{n-2}(\tilde{k}^{n-3})$ not depending on $k_{n-2}$, and substitute (\ref{nonstationary:example:gaussian:equation120}) into the integral inequality in (\ref{nonstationary:example:gaussian:equation118}) to obtain
\begin{align*}
\alpha_{n-2}(\tilde{k}^{n-3}){e}^{-\bar{g}_{n-2,n}(\tilde{k}^{n-3})}\int_{{\cal K}_{n-2}}{e}^{s||k_{n-2}-\tilde{k
}_{n-2}||_{2}^2}dk_{n-2}\leq{1}.
\end{align*}
By change of variable of integration then
\begin{align*}
\alpha_{n-2}(\tilde{k}^{n-3}){e}^{-\bar{g}_{n-2,n}(\tilde{k}^{n-3})}&\int_{-\infty}^{\infty}{e}^{s||z_{n-2}||_{2}^2}dz_{n-2}=\alpha_{n-2}(\tilde{k}^{n-3}){e}^{-\bar{g}_{n-2,n}(\tilde{k}^{n-3})}\left(-\frac{\pi}{s}\right)^{\frac{p}{2}}\leq{1}.
\end{align*}
Hence,
\begin{align}
\alpha_{n-2}(\tilde{k}^{n-3})\left(-\frac{\pi}{s}\right)^{\frac{p}{2}}\leq{e}^{\bar{g}_{n-2,n}(\tilde{k}^{n-3})}.\label{nonstationary:example:gaussian:equation121}
\end{align}
\noi Moreover, $\alpha_{n-2}(\cdot)$ is chosen so that the inequality in (\ref{nonstationary:example:gaussian:equation121}) holds with equality, giving 
\begin{align}
\alpha_{n-2}(\tilde{k}^{n-3})&=\frac{{e}^{\bar{g}_{n-2,n}(\tilde{k}^{n-3})}}{\int{e}^{s||z_{n-2}||_{2}^2}dz_{n-2}}={e}^{\log{\alpha_{n-1}(\tilde{k}^{n-2})}+H(K_{n-1}|\tilde{K}^{n-3}=\tilde{k}^{n-3})}\left(-\frac{s}{\pi}\right)^{\frac{p}{2}}\nonumber\\
&=\left\{\left(-\frac{s}{\pi}\right)^{\frac{p}{2}}\right\}^{3} e^{H(K_n|\tilde{K}^{n-3}=\tilde{k}^{n-3})+H(K_{n-1}|\tilde{K}^{n-3}=\tilde{k}^{n-3})}.\label{nonstationary:example:gaussian:equation122}
\end{align}
Therefore, (\ref{nonstationary:example:gaussian:equation120}) is given by
\begin{align}
\lambda_{n-2}(k_{n-2},\tilde{k}^{n-3})&=\frac{\left\{\left(-\frac{s}{\pi}\right)^{\frac{p}{2}}\right\}^{3} e^{H(K_n|\tilde{K}^{n-3}=\tilde{k}^{n-3})+H(K_{n-1}|\tilde{K}^{n-3}=\tilde{k}^{n-3})}}{\bar{p}(k_{n-2}|\tilde{k}^{n-3})}.\label{nonstationary:example:gaussian:equation123}
\end{align}
Substituting (\ref{nonstationary:example:gaussian:equation123}) into term-(n-2) of (\ref{nonstationary:example:gaussian:equation108}) gives
\begin{align}
\underline{Term-(n-2):}\stackrel{(h)}\geq&-\int_{{\cal K}_{n-2}\times\tilde{\cal K}^{n-3}}\Big(\int_{\tilde{\cal K}_{n-2}}\bar{g}_{n-2,n}(\tilde{k}^{n-3})\bar{p}(\tilde{k}_{n-2}|\tilde{k}^{n-3},k_{n-2})d\tilde{k}_{n-2}\Big)\nonumber\\
&\bar{p}(k_{n-2},\tilde{k}^{n-3})dk_{n-2}d\tilde{k}^{n-3}\nonumber\\
&+\int_{{\cal K}_{n-2}\times\tilde{\cal K}^{n-3}}\log\Big(\lambda_{n-2}(k_{n-2},\tilde{k}^{n-3})\Big)\bar{p}(k_{n-2},\tilde{k}^{n-3})dk_{n-2}d\tilde{k}^{n-3}\label{nonstationary:example:gaussian:equation122a}\\
&\stackrel{(i)}=-\frac{1}{n+1}\log\left\{\left(-\frac{s}{\pi}\right)^{\frac{p}{2}}\right\}^{2}-H(K_n|\tilde{K}^{n-3})-H(K_{n-1}|\tilde{K}^{n-3})\nonumber\\
&+\frac{1}{n+1}\int_{{\cal K}_{n-2}\times\tilde{\cal K}^{n-3}}\log\Big(\alpha_{n-2}(\tilde{k}^{n-3})\Big)\bar{p}(k_{n-2},\tilde{k}^{n-3})dk_{n-2}d\tilde{k}^{n-3}\nonumber
\end{align}
\begin{align}
&-\frac{1}{n+1}\int_{{\cal K}_{n-2}\times\tilde{\cal K}^{n-3}}\log\Big(\bar{p}(k_{n-2}|\tilde{k}^{n-3})\Big)\bar{p}(k_{n-2},\tilde{k}^{n-3})dk_{n-2}d\tilde{k}^{n-3}\nonumber\\
&=-\frac{1}{n+1}\log\left\{\left(-\frac{s}{\pi}\right)^{\frac{p}{2}}\right\}^{2}-\frac{1}{n+1}H(K_n|\tilde{K}^{n-3})\nonumber\\
&-\frac{1}{n+1}H(K_{n-1}|\tilde{K}^{n-3})+\frac{1}{n+1}\log\left\{\left(-\frac{s}{\pi}\right)^{\frac{p}{2}}\right\}^{3}\nonumber\\
&+H(K_n|\tilde{K}^{n-3})+\frac{1}{n+1}H(K_{n-1}|\tilde{K}^{n-3})+\frac{1}{n+1}H(K_{n-2}|\tilde{K}^{n-3})\nonumber\\
&=\frac{1}{n+1}\log\left(-\frac{s}{\pi}\right)^{\frac{p}{2}}+\frac{1}{n+1}H(K_{n-2}|\tilde{K}^{n-3}) \label{nonstationary:example:gaussian:equation122b} ,
\end{align}
where $(h)$ follows from the fact that $g_{n-2,n}(\tilde{k}^{n-2})\leq\bar{g}_{n-2,n}(\tilde{k}^{n-3})$ (see (\ref{nonstationary:example:gaussian:equation116c})), and $(i)$ follows by substituting (\ref{nonstationary:example:gaussian:equation116c}) and (\ref{nonstationary:example:gaussian:equation120}) into the the second and third expression of (\ref{nonstationary:example:gaussian:equation122a}), respectively.\\
\noi By applying induction, we obtain the following lower bound for the NRDF.
\begin{align}
R_{0,n}^{na,K^n,\tilde{K}^n}(D)&\geq sD+\frac{1}{n+1}\left\{\left(-\frac{s}{\pi}\right)^{\frac{p}{2}}\right\}^{n+1}\nonumber\\
&+\frac{1}{n+1}\Big\{H(K_n|\tilde{K}^{n-1})+H(K_{n-1}|\tilde{K}^{n-2})+\ldots+H(K_1|\tilde{K}_0)+H(K_0)\Big\}\nonumber\\
&=sD+\frac{1}{2}\frac{1}{n+1}\sum_{t=0}^n\sum_{i=1}^p\log\left(-\frac{s}{\pi}\right)+\frac{1}{n+1}\sum_{t=0}^n{H}(K_t|\tilde{K}^{t-1})\nonumber\\
&\stackrel{(j)}=sD+\frac{1}{2}\frac{1}{n+1}\sum_{t=0}^n\sum_{i=1}^p\log\left(-\frac{s}{\pi}\right)+\frac{1}{2}\frac{1}{n+1}\sum_{t=0}^n\log{2}\pi{e}|\Lambda_t|\label{nonstationary:example:gaussian:equation126} ,
\end{align}
where $(j)$ follows from the fact that 
\begin{align*}
{H}(K_t|\tilde{K}^{t-1})&=H(X_t-\mathbb{E}\big\{X_t|\sigma\{K^{t-1}\}\big\}\big{|}\tilde{K}^{t-1})\nonumber\\
&=H(X_t|\tilde{K}^{t-1})=H(X_t)=\frac{1}{2}\sum_{t=0}^n\log{2}\pi{e}|\Lambda_t|.
\end{align*}
\noi Next, we show how to find the Lagrangian multiplier $``s"$ so that the lower bound (\ref{nonstationary:example:gaussian:equation126}) equals $\frac{1}{2}\sum_{t=0}^n\sum_{i=1}^p\log\left(\frac{\lambda_{t,i}}{\delta_{t,i}}\right)$. To this end, we need to ensure existence of some $s<0$ such that the following identity holds.
\begin{align}
&sD+\frac{1}{2}\frac{1}{n+1}\sum_{n=0}^n\sum_{i=1}^p\log\left(-\frac{s}{\pi}\right)+\frac{1}{2}\frac{1}{n+1}\sum_{t=0}^n\log{2}\pi{e}|\Lambda_t|=\frac{1}{2}\frac{1}{n+1}\sum_{t=0}^n\sum_{i=1}^p\log\left(\frac{\lambda_{t,i}}{\delta_{t,i}}\right).\nonumber
\end{align}
After some algebra, the previous expression can be  simplified into the following expression.
\begin{align}
&\frac{1}{2}\log{e}^{2s\frac{1}{(n+1)}\sum_{t=0}^n\sum_{i=1}^p\delta_{t,i}}+\frac{1}{2}\frac{1}{n+1}\sum_{t=0}^n\sum_{i=1}^p\log\left(-\frac{s}{\pi}\right)=\frac{1}{2}\frac{1}{n+1}\sum_{t=0}^n\sum_{i=1}^p\log\frac{1}{2\pi{e}\delta_{t,i}}\nonumber .
\end{align}
In turn, from the equation above we obtain 
\begin{align}
\frac{1}{2}\frac{1}{n+1}\sum_{t=0}^n\sum_{i=1}^p\log{e}^{2s\delta_{t,i}}\left(-\frac{s}{\pi}\right)=\frac{1}{2}\frac{1}{n+1}\sum_{t=0}^n\sum_{i=1}^p\log\frac{1}{2\pi{e}\delta_{t,i}}\Longrightarrow{\delta_{t,i}}=-\frac{1}{2s}\nonumber ,
\end{align}
where $\delta_{t,i}=\{\xi,\lambda_{t,i}\}$. Now, if $\delta_{t,i}=\xi$ then ${\delta_{t,i}}=-\frac{1}{2s}$ and the NRDF is bounded below by the following expression
\begin{align}
R_{0,n}^{na,K^n,\tilde{K}^n}(D)\geq\frac{1}{2}\frac{1}{n+1}\sum_{t=0}^n\sum_{i=1}^p\log\left(\frac{\lambda_{t,i}}{\delta_{t,i}}\right),~\frac{1}{n+1}\sum_{t=0}^n\sum_{i=1}^p\delta_{t,i}=D.\nonumber
\end{align}
{\bf (2)} The estimation error $\widehat{X}_{t|t-1}$ is given by the modified Kalman filter equations  (\ref{nonstationary:example:gaussian:10})-(\ref{nonstationary:example:gaussian:11i}) (see \cite[Theorem 1.1, pp. 158]{caines1988}). Note that (\ref{nonstationary:example:gaussian:11i}) is computed as follows. 
\begin{align}
\begin{split}
M_t&=E_t^{\T}H_t{E}_{t}\Pi_{t|t-1}(E_{t}^{\T}H_{t}E_{t})^{\T}+E_{t}^{\T}{\Phi}_{t}Q_t{\Phi}_{t}^{\T}E_t\\
&\stackrel{(a)}=E_t^{\T}H_t{\Lambda_t}H_{t}E_{t}+E_{t}^{\T}H_t\Delta_t{E}_t =E_t^{\T}H_t{\Lambda_t}{E}_t\label{proof:fully:1} ,
\end{split}
\end{align}
where $(a)$ follows if by setting ${\Phi}_t=\sqrt{H_t\Delta_t{Q}_t^{-1}}$. By substituting (\ref{nonstationary:example:gaussian:11i}) into \eqref{nonstationary:example:gaussian:11} we obtain
\begin{align}
\Pi_{t+1|t}=A_tE^{\T}_t\Delta_tE_tA_t^{\T}+B_tB_{t}^{\T}.\label{proof:fully:2}
\end{align}
{\bf (3)} Next, we determine the realization of the optimal reproduction distribution. Recall that $\Pi_{t|t-1}$ is given by \eqref{nonstationary:example:gaussian:equa.133i}. Therefore, to determine $\Pi_{t|t-1}$, we need the equation of the error $e_t\triangleq{X}_t-\widehat{X}_{t|t-1}$, hence the equation of the least-squares filter of $X_t$ given all the previous outputs $Y^{t-1}$, namely $\widehat{X}_{t|t-1}$. From Fig.~\ref{nonstationary:communication_system}, we deduce that ${Y}_t=\tilde{K}_t+\widehat{X}_{t|t-1}$, where $\{\widehat{X}_{t|t-1}:~t\in\mathbb{N}_0\}$ is obtained from the modified Kalman filter $\widehat{X}_{t|t-1}$. Thus, we obtain \eqref{nonstationary:10ab}. \\
{\bf (4)} By substituting \eqref{proof:fully:1} in \eqref{nonstationary:example:gaussian:10} we obtain the updated version of $\widehat{X}_{t|t-1}$ as follows. 
\begin{align}
\begin{split}
&\widehat{X}_{t+1|t}=A_t\widehat{X}_{t|t-1}+A_t\Pi_{t|t-1}(E_t^{\T}H_{t}E_{t})^{\T}M_t^{-1}\left({Y}_t-\widehat{X}_{t|t-1}\right)\\
&=A_t\widehat{X}_{t|t-1}+A_t\Pi_{t|t-1}{E}_t^{\T}H_{t}E_{t}{E}^{\T}_t\diag\{\frac{1}{\lambda_{t,1}},\ldots,\frac{1}{\lambda_{t,p}}\}H^{-1}_t{E}_t\left({Y}_t-\widehat{X}_{t|t-1}\right)\\
&=A_t\widehat{X}_{t|t-1}+A_t{E}^{\T}_t\Lambda_tE_t{E}_t^{\T}H_{t}\diag\{\frac{1}{\lambda_{t,1}},\ldots,\frac{1}{\lambda_{t,p}}\}H^{-1}_t{E}_t\left({Y}_t-\widehat{X}_{t|t-1}\right)=A_tY_{t}.
\label{proof:fully:3}
\end{split}
\end{align}
Using \eqref{proof:fully:3}, we obtain \eqref{filter:outcome:eq:1} and since $\widehat{X}_{t+1|t}=A_t\widehat{X}_{t|t}$ we also obtain \eqref{filter:outcome:eq:2}.
\par Finally, by substituting \eqref{filter:outcome:eq:1} in \eqref{nonstationary:10ab} we obtain \eqref{nonstationary:10abbba}.\\
\noi{\bf (5)} To show the last stage of our theorem, we note that ${\cal V}_t\triangleq{Y}_t-\mathbb{E}\left\{Y_t|\sigma\{Y^{t-1}\}\right\}$ is the innovation process of \eqref{nonstationary:10ab}, and that ${\cal V}_t\triangleq{Y}_t-\widehat{X}_{t|t-1}\equiv\tilde{K}_t$. Moreover, since 
$\tilde{K}_t=E_t^{\T}\Phi_tZ_t$ and $\{E_t^{\T}\Phi_t:~t\in\mathbb{N}_0^n\}$ are invertible, then the statement holds. This completes the proof.

\bibliographystyle{ieeetr}
\bibliography{photis_references_nonstationary_rdf}

%
%
%
%
\end{document}